%% file: main.tex
\documentclass[12pt]{article}
\usepackage
[
    a4paper,
    margin = 3cm,
]
{geometry}
\usepackage{amsfonts}
\usepackage{amsmath}
\usepackage{amssymb}
\usepackage{amsthm}
\usepackage{graphicx}
\usepackage{tikz}
\usepackage{sgame}
\usetikzlibrary{decorations.pathreplacing}
\usepackage{color}
\usepackage{mathpazo}
\usepackage{times}
\usepackage{caption}
\usepackage{subcaption}
\usepackage{natbib}
\usepackage[labelsep=space]{caption}
\usepackage{soul}
\usepackage{url}
\usepackage[multiple]{footmisc}
\usepackage{comment}

\theoremstyle{definition}

\newtheorem{corollary}{Corollary}
\newtheorem{definition}{Definition}

\newtheorem{lemma}{Lemma}
\newtheorem{proposition}{Proposition}

\newcommand\citeapos[1]{\citeauthor{#1}'s (\citeyear{#1})}

\DeclareMathOperator\erf{erf}

\DeclareFontFamily{U}{mathx}{\hyphenchar\font45}
\DeclareFontShape{U}{mathx}{m}{n}{
      <5> <6> <7> <8> <9> <10>
      <10.95> <12> <14.4> <17.28> <20.74> <24.88>
      mathx10
      }{}
\DeclareSymbolFont{mathx}{U}{mathx}{m}{n}
\DeclareFontSubstitution{U}{mathx}{m}{n}
\DeclareMathAccent{\widecheck}{0}{mathx}{"71}

\newcommand{\RR}{\mathbb{R}}
\newcommand{\NN}{\mathbb{N}}
\newcommand{\EE}{\mathbb{E}}
\newcommand{\dd}{\mathrm{d}}
\newcommand{\ee}{\mathrm{e}}
\newcommand{\Action}{\mathcal{A}}
\newcommand{\vl}{\,\vert\,}

\newcommand{\prob}{\mathrm{Pr}}
\newcommand{\invest}{\mathcal{E}}
\newcommand{\notinvest}{\mathcal{N}}

\newcommand{\one}{\mathbb{1}}

\newcommand{\thetazerostar}{\theta_L^*}

\newcommand{\thetazerohat}{\widehat{\theta}_L}
\newcommand{\xhhat}{\widehat{x}_h}

\newcommand{\xistar}{x_{\invest}^*}
\newcommand{\xnstar}{x_{\notinvest}^*}

\newcommand{\thetalow}{\underline{\theta}_L}
\newcommand{\thetaup}{\overline{\theta}_L}

\newcommand{\xilow}{\underline{x}_\invest}
\newcommand{\xiup}{\overline{x}_\invest}

\newcommand{\xnlow}{\underline{x}_\notinvest}
\newcommand{\xnup}{\overline{x}_\notinvest}

\newcommand{\xhlow}{\underline{x}_h}
\newcommand{\xhup}{\overline{x}_h}

\newcommand{\br}{\mathrm{BR}}
\newcommand{\xlow}{\underline{x}_L}
\newcommand{\xup}{\overline{x}_L}

\begin{document}

\title{It's Not Always the Leader's Fault: How Informed Followers Can Undermine Efficient Leadership\thanks{~An earlier draft of this paper was titled ``The Signaling Role of Leaders in Global Games". We owe special thanks to Alessandro Pavan for his support and guidance throughout the writing process. We also wish to thank George-Marios Angeletos, Sandeep Baliga, Georgy Egorov, Stephen Morris, Wojciech Olszewski, Harry Pei, Alvaro Sandroni, Marciano Siniscalchi, Zhen Zhou as well as conference participants at the 34th Stony Brook International Conference on Game Theory, the C.R.E.T.E. 2022 conference and various seminar audiences at Northwestern University for valuable comments and suggestions.}} 

\author{Panagiotis Kyriazis\thanks{~Corresponding author. Department of Economics, Northwestern University. Email: pkyriazis@u.northwestern.edu.}
\and Edmund Y. Lou\thanks{~Department of Economics, Northwestern University. Email: edmund.lou@u.northwestern.edu.}}
\date{\today}

\maketitle
\thispagestyle{empty}

\begin{abstract}
Coordination facilitation and 
efficient decision-making are two essential components of successful leadership. In this paper, we take an informational approach and investigate how followers' information impacts coordination and efficient leadership in a model featuring a leader and a team of followers. We show that efficiency is achieved as the unique rationalizable outcome of the game when followers possess sufficiently imprecise information. In contrast, if followers have accurate information, the leader may fail to coordinate them toward the desired outcome or even take an inefficient action herself. We discuss the implications of the results for the role of leaders in the context of financial fragility and crises. \\


\end{abstract}

\newpage
\pagenumbering{arabic}

\input{intro}

\medskip

\input{model}

\medskip

\input{equilibrium}

\medskip

\input{extension}

\medskip

\input{conclusion}
\newpage

\bibliographystyle{te}
\bibliography{references}

\newpage

\input{appendix_a}

\newpage
\input{appendix_b}

\newpage
\input{appendix_c}

\end{document}

%% file: intro.tex
\section{Introduction}

Coordination problems are pervasive in a wide range of
socio-economic phenomena, including but not limited to corporate culture, bank runs, currency attacks and political change. 
In many of these scenarios, a leader is responsible for taking the correct action and incentivizing the followers to do the same.
This task is often challenging due to the inherent uncertainty of certain economic environments and the positive spillover effect of coordination that typically benefits everyone involved more than each actor separately. 
Leaders can potentially succeed in their role by \textit{leading by example}. Exerting costly effort or taking a risky action can signal to the followers the high return potential of a given course of action. 
However, this may not be enough to incentivize the followers to take the correct action due to the dispersed information about the economic environment that creates strategic uncertainty about what each follower will eventually do. This strategic uncertainty that agents face concerning the actions and beliefs of others may lead to undesirable outcomes. It is, therefore, natural to ask under what conditions coordination can be achieved and what are the specific characteristics of the leader and the followers that facilitate or undermine it. This paper explores this question from an informational perspective and argues that precise information held by followers can potentially undermine a leader's ability to coordinate the followers on a mutually desirable course of action.

Our framework focuses on a leader who is perfectly informed about the state of nature, and a team of followers who have access to private information. The leader leads by example---she makes the first move, which the followers observe.
Then, they simultaneously make their choices, incorporating the information that becomes available due to the leader's action. Our model, thus, features signaling on the part of the leader. It is important to note that the leader is not special in any other way besides acting first and having an informational advantage. We do not consider issues of leader's credibility or special characteristics that might make her a better or worse leader. We take this route in order to isolate the effect that is solely due to information. Our main finding is that if the followers' information is precise enough, efficient leadership may not be attainable. More importantly, the leader might act inefficiently herself due to fear of miscoordination by the followers. Conversely, if the followers have sufficiently imprecise information, efficient coordination can be sustained as the unique rationalizable outcome.

Formally, our game is a two-stage game with binary, irreversible actions that features strategic complementarities both within and across stages. We use $\Delta$-rationalizability as our solution concept, which extends extensive-form rationalizability \`{a} la \cite{pearce_1984} to games with incomplete information \citep{battigalli_siniscalchi_2003}.  We fully characterize the rationalizable set for the leader and followers, that is, the set of strategies that are rationalizable for certain types of each player. Our main results identify necessary and sufficient conditions about the noise of followers' information that are responsible for uniqueness or multiplicity of rationalizable play. When rationalizable play is unique, the leader succeeds in inducing the followers to take the correct action. When this is not the case, followers might play against the leader which in turn may make the leader choose an undesirable action despite being perfectly informed about the state. We note that our results would also obtain under alternative solution concepts such as interim sequential rationalizability \citep{penta_2012} or interim correlated rationalizability applied to the normal form of the game \citep{dekel_et_al_2007, chen_2012}. 

Our main result is driven by the tension between two opposing effects: the signaling effect of the leader's choice on the followers and the miscoordination effect that arises from the followers' dispersed information.
The leader's action reveals some information to the followers and can influence their decisions.
However, each follower is uncertain about which leader's type chose the observed action and/or about the types and actions of other followers, which leads to the miscoordination effect.
Knowing this, the leader may choose the undesired action in the first place. The signaling effect is captured by the truncation of the conditional distribution of the state that followers deem possible after observing the leader's choice. The miscoordination effect is captured by the conditional higher-order beliefs a follower holds about the state and the beliefs of other followers. More formally, this effect is captured by the \textit{conditional rank belief function}, a generalization of the rank belief function introduced in \cite{morris_et_al_2016} and \cite{morris_yildiz_2019}. This function yields the probability that a follower assigns to the event that another follower has received a signal about the state less than his own, which, in our environment, depends on the behavior of the leader. Specifically, we show that the conditions that ensure unique rationalizable play essentially amount to a bound on the derivative of the conditional rank belief function that guarantees that the signaling effect dominates the miscoordination effect.

Dominance of the signaling effect over the miscoordination effect means that each follower wants to imitate the leader irrespective of his signal and his beliefs about the actions of other followers. A necessary condition for this to be the case is that the leader is well-informed. If this is not true, then a follower may worry about potential miscoordination. The tension comes because both these forces depend on the leader's behavior but in opposite ways: more ``aggressive'' (and, thus, more efficient in the case where the state is positive) behavior by the leader weakens the signaling effect and strengthens the miscoordination effect by increasing the strategic uncertainty that followers face. This effect undermines coordination on the correct course of action on the part of followers. The only way it can be overcome, so that the efficient action of the leader and followers is supported as the unique rationalizable behavior, is for the noise in the followers' information to be sufficiently high. If this is not the case, the miscoordination effect dominates: even if it is common knowledge which action is the efficient one given the state, followers may still choose not to take this action. That inefficient actions for the followers are rationalizable makes the choice of the leader not to take the efficient action in the first place rationalizable as well. In the language of financial economics, this inefficient outcome may prevail solely due to panic, that is, solely due to the fear of miscoordination and not due to fundamentals being weak. 

We extend our results to a setting where the leader is not perfectly informed about the state but has access to private information like the followers. The novel effect is that this uncertainty “noises up” the learning induced by the action of the leader: whereas in the benchmark model this knowledge leads to a truncation in the support of posterior beliefs about the state, now posterior beliefs retain full support over the entire real line. This extension, which from an applied point of view is clearly more natural, establishes the generality of our result because it restores two-sided dominance, a property that is absent in our main model but is key to the global games literature, to which this paper broadly belongs. 

While our model is very stylized, it plays the role of a metaphor that can be applied to a variety of situations. As we stated, it features strategic complementarities and does not allow for free riding, something that is usually central in organizational settings or public good provision games. However, by excluding this issue, we can analyze more diverse phenomena, in which there might not be a leader, in the standard use of the term, but a player whose action is visible to smaller players. What is important for our results, is the visibility of this player's action that grants them the leadership role. The leader, thus, can be a nation that initiates an environmentally friendly policy in the hope that more countries will follow or the manager of a firm that wishes to induce her employees to exert costly effort when her only instrument is her own choice to work or shirk. It can also be a prominent investor contemplating whether to attack a currency or not or whether to roll over debt or not. Or, the leader can be a vanguard in revolution. By attacking the regime first, she can inspire other citizens to do so. As we point out, the (application-specific) ideal outcome from the point of view of the leader is supported as the unique rationalizable outcome only when the conditions we identify are satisfied. As an example, we discuss the implications of the forces involved in our model for financial stability and panic-driven crises. Specifically, we compare the predictions of our model to results in the finance literature according to which the leader can achieve her desired outcome. This need not be the case in our framework. In particular, when followers' information is precise, followers may coordinate against the leader, resulting in a negative payoff for her. This makes the presence of the leader less influential, in the sense that she may not take the action that (de)stabilizes the economy even if the fundamentals perfectly justify doing so.
Of course, our results are not to be interpreted as suggesting that followers should always be ``kept in the dark.'' While this facilitates coordination and efficient leadership in our setting, there are many situations in which the leader might want to learn from the followers.

\subsection{Related Literature}

Our paper contributes to two strands of literature. Primarily, it adds to the literature on global games \citep{carlsson_van_damme_1993, morris_shin_2003} by characterizing rationalizable outcomes when one of the players moves first and her action is observable by the others. \cite{corsetti_et_al_2004} is the most closely related paper as they investigate a large player's role in currency attacks. Although the settings are different, our results complement theirs in identifying the precise conditions under which efficient coordination is the unique rationalizable outcome. We discuss in detail how our paper relates to \cite{corsetti_et_al_2004} in Section 6. 

\cite{loeper_et_al_2014} study the role of experts whose actions can influence agents' behavior and show that the outcome is biased towards experts' interests. This feature is also present in our paper when there is unique rationalizable play. The difference between the two models is that in \cite{loeper_et_al_2014}, the experts' preferences differ from those of the agents, whereas in our model, the leader and followers share the same objective. Moreover, experts' actions in \cite{loeper_et_al_2014} have no spillover effect and do not affect directly the outcome of the game. Finally, while the main question of interest in \cite{loeper_et_al_2014} is how experts influence agents' actions, we ask how information on the part of followers affects coordination and the choice of the leader.
In \cite{angeletos_et_al_2006} and \cite{angeletos_pavan_2013} there is a perfectly informed policymaker whose action is observable by the players before the coordination stage so the game also features signaling. In their framework, this leads to equilibrium multiplicity. The difference between the two papers is that the leader and the followers have perfectly aligned incentives in our model, while in \cite{angeletos_et_al_2006} and \cite{angeletos_pavan_2013} there is a conflict of interest over a subset of the state space. \cite{basak_zhou_2020} study a regime change game where agents choose sequentially and there is a principal who can dynamically disclose information to dissuade them from attacking. The disclosure policies studied also remove the two-sided dominance property. However,  in our model, the leader and followers share the exact same goal and, moreover, we do not have disclosure of information but signaling by the leader. \cite{angeletos_et_al_2007} extend the standard static benchmark of a regime change game to a dynamic setting in which players may attack a regime multiple times and where learning features a prominent role. They derive a multiplicity result similar to ours, due to the role of the truncation of conditional beliefs that is central to both papers. \cite{huang_2017} incorporates the policy maker's reputation into the model of \cite{angeletos_et_al_2007} and shows that whether an equilibrium with attacks exists or not depends on the speed of learning by the agents of the policy maker's type. However, in both these papers, agents' behavior in the continuation does not affect how they act when the game begins, while, in our setting, followers' behavior in the sub-game obviously affects the decision of the leader at the beginning of the game.

Moreover, our paper is related to the economics of leadership literature pioneered by \cite{hermalin_1998}.
Our work is most closely related to  \cite{komai_et_al_2007} and \cite{komai_stegeman_2010}.
The former examines the informational role of the leader in a moral hazard in teams model, while the latter extends this setting by exploring leader selection and investment in information.
 Our results validate those in \cite{komai_et_al_2007} and \cite{komai_stegeman_2010} in that we also derive efficient outcomes when the leader has a sufficient informational advantage and reveals, through signaling, part of her information to followers. 
However, our contribution lies in the converse direction: whenever the leader does not have a sufficient informational advantage, then the leader herself may act inefficiently. In addition, as our extension shows, the result does not necessarily rely on who is better informed but rather on whether the signaling effect of the leader's action dominates the miscoordination effect inherent to the followers' problem due to dispersed information. Other papers in the economics of leadership literature that relate to our work include \cite{bolton_et_al_2013} and \cite{dewan_myatt_2008}. 
These papers explore in a beauty contest framework, the qualities a leader must possess to successfully determine an organization's mission in changing environments while ensuring coordination among followers. Our paper differs with respect to \cite{bolton_et_al_2013} and \cite{dewan_myatt_2008} in that, in our model, the leader neither chooses the signals that followers observe as in \cite{dewan_myatt_2008} nor has the final say about which action will be taken as in \cite{bolton_et_al_2013}. Her role is simpler: she just chooses an action that the followers observe. By simplifying the environment in this way, we are able to isolate and analyze the effect of the information that followers possess on leadership in a clearer manner: the leader may choose the inefficient action and potentially trap herself and the followers in the wrong course of action, a feature absent in \cite{bolton_et_al_2013} and \cite{dewan_myatt_2008}.

On the methodological side, our model is an extensive-form game with incomplete information that features strategic complementarities within and across stages. It, therefore, shares elements with \cite{echenique_2004} and \cite{van_zandt_vives_2007}. Moreover, the paper is related to \cite{morris_et_al_2016} and \cite{morris_yildiz_2019}. Even though we ask different questions, we also use properties of rank belief functions to analyze the behavior of followers after observing the leader's action.

\subsection{Organization of the Paper}

The remainder of the paper is organized as follows: Section 2 introduces the main model. In Section 3, we present and discuss the main results. Section 4 extends the results of Section 3 to an alternative information structure for the leader. In Section 5, we explore the implications of the results for the leader's role in financial stability and crises. Section 6 discusses the choice of the solution concept. Finally, Section 7 concludes. All proofs are in Appendix A.

%% file: model.tex
\section{The Model} \label{sect_model}

Consider the following two-stage game. There is a \textit{leader} (she), $L$, and a team of $n \geq 2$ \textit{followers} (he). Each player $i \in N \equiv \{L, 1, \dots, n\}$ has to decide whether to take an action ($a_i = \invest$) or not ($a_i = \notinvest$). This action can be interpreted as exerting costly effort, investing into a new project, or attacking a regime or currency. The cost of taking the action is $c > 0$. 
Not taking the action is a safe option (e.g., shirking or staying with the existing technology) with benefits normalized to zero. Let $\Action_i = \{\invest, \notinvest\}$ be the set of actions for player $i \in N$. To fix ideas, we will henceforth refer to the action as the exertion of effort.

The leader moves first in stage 1. The followers, $j \in F = \{1, \dots, n\}$, make decisions in stage 2 after observing the leader's action. Our game thus constitutes a multi-stage game with observable actions.
We assume that the leader's action is irreversible; that is, she can neither ``exit'' nor ``delay'', which may be understood as the commitment made by the leader or the consequence of a high reputation cost of exit or delay.\footnote{~We borrow the terms from \cite{kovac_steiner_2013}.} Thus, our leader has maximum credibility.

We consider the situation where coordination exhibits positive externalities; that is, the more people choose to exert effort the more benefit everybody accrues. In other words, players'  decisions are strategic complements. Let $\tilde{b}(\theta, a_{-i})$ be an additively separable benefit function for player $i$ when action $a_i = \invest$ is taken (otherwise it is zero). The parameter $\theta \in \Theta = \RR$ is a payoff-relevant state (``fundamentals'') that affects the gross return. The state $\theta$ can be though of as a parametrization of the environment in which an organization operates, the strength of a regime, or the solvency of a bank.  Let $a_{-i} = (a_k)_{k \in N; \,k \neq i}$ be the vector of other players' actions that generates positive \emph{spillover benefits} when at least one other player chooses to exert effort. Thus, player $i$'s payoff is given by $u_i = \tilde{b}(\theta, a_{-i}) - c$. 
We further assume that $\tilde{b}$ is strictly increasing in both $\theta$ and $a_{-i}$ and 
is symmetric in $a_{-i}$. Let $\one(\cdot)$ be the indicator function and $A_{-i} = n^{-1}\sum_{k \neq i,~k \in N} \one(a_k = \invest)$ be the proportion of other players choosing to exert effort. Then we could write player $i$'s payoff as\footnote{~Suppose that $\tilde{b}(\theta, a_{-i}) = v(\theta) + \sum_{k \neq i}w(a_k)$, where $v$ is an increasing function and $w$ is such that $w(0) = 0$ and $w(d_k) = \omega > 0$. Then player $i$'s payoff is $\tilde{u}(\theta, a_{-i}) = \tilde{b}(\theta, a_{-i}) - c = v(\theta) + \omega \sum_{k \neq i}\one(a_k = \invest) - c$. A monotone transformation, $u = \alpha \tilde{u} + \beta$, gives that $u(\theta, A_{-i}) = \alpha v(\theta) + \beta + A_{-i} - 1$ by letting $\alpha = 1/(\omega n)$ and $\beta = 1 - \alpha c$. We may then assume, without loss of generality, that $\alpha v + \beta = \mathrm{id}$, the identity map from $\RR$ to $\RR$.}
\begin{equation*}
    u_i = u(\theta, A_{-i}) = \theta + A_{-i} - 1.
\end{equation*}
This payoff function is familiar in the global games literature, for example, see \cite{morris_shin_2003} and \cite{morris_yildiz_2019}. Figure \ref{fig:example} illustrates an example with two followers ($n = 2$).

\input{gametree}

Assume that the leader observes the realization of $\theta$ (henceforth her type). Followers, on the other hand,  have \textit{fundamental uncertainty} about $\theta$. The initial common prior is an improper uniform distribution over the real line.\footnote{~Later we discuss why this assumption, though simplifying, actually strengthens our results.} Each follower $j \in F$ receives a signal $x_j = \theta + \sigma_F \varepsilon_j$, where $\sigma_F > 0$ measures the quality of private information and $\varepsilon_j$ is an idiosyncratic standard Gaussian noise that is independent of $\theta$ and independently and identically distributed (IID) across all followers. In Appendix C we show that our results hold for a set of noise distributions with positive densities, continuously differentiable, symmetric (around zero), and log-concave over $\RR$. We refer to $x_j$ as follower $j$'s type and let $X_j = \RR$ be the corresponding type space.

From a global game perspective, the assumption that the leader knows $\theta$ is common knowledge is critical. In particular, under this assumption, the subgames do not have two-sided dominance regions, but rather there is only one-sided dominance. In Section 4, we investigate an alternative information structure where the leader also observes a noisy private signal about $\theta$, which brings back the standard two-sided dominance. The current model can be understood as the limiting case when the noise of the leader's information approaches zero while keeping $\sigma_F$ fixed. The spirit of the main result is similar in both cases.

Note that, under complete information, it is straightforward to see that the model admits multiple subgame perfect equilibria when $\theta \in (0, 1/n)$. In the two-follower case, for example, we have two subgame perfect equilibria---$(\invest, \invest\notinvest, \invest\notinvest)$ and$(\notinvest, \notinvest\notinvest, \notinvest\notinvest)$---with the former being fully efficient.\footnote{~By $\invest\notinvest$ we mean a follower exerts effort when $a_L = \invest$ and does not exert effort when $a_L = \notinvest$.}

%% file: gametree.tex
\begin{figure}[!htp]
    \centering
    \begin{tikzpicture}[scale=0.365]
        \tikzstyle{solid node}=[circle,draw,inner sep=1.2];
        \tikzstyle{hollow node}=[circle,draw, inner sep=1.2];
        \tikzstyle{level 1}=[level distance=20mm,sibling distance=210mm]
        \tikzstyle{level 2}=[level distance=15mm,sibling distance=25mm]
        \renewcommand{\gamestretch}{1.2}
        \node(0)[hollow node]{}
        child{node{}
        edge from parent node[above left]{$a_L = \invest$}
        }
        child{node{}
        edge from parent node[above right]{$a_L = \notinvest$}
        };
        
        \node[circle, fill=black, inner sep=1.2, label=above:{\text{Leader}}]at(0){};
        \node[below, xshift=0, yshift=10]at(0-1){
        \arrayrulewidth.75pt
        \begin{game}{2}{2}
        & $\invest$   & $\notinvest$\\
        $\invest$ & $\theta~,  ~\theta~, ~\theta$ & $\theta-\frac{1}{2}~, ~\theta-\frac{1}{2}~, ~0$\\
        $\notinvest$ & $\theta-\frac{1}{2}~, ~0~, ~\theta-\frac{1}{2}$ & $\theta-1~, ~0~, ~0$
        \end{game}
        };
        \node[below, xshift=-5, yshift=10]at(0-2){
        \arrayrulewidth.75pt
        \begin{game}{2}{2}
        & $\invest$   & $\notinvest$\\
        $\invest$ & $0~, ~\theta-\frac{1}{2}~, ~\theta-\frac{1}{2}$ & $0~, ~\theta-1~, ~0$\\
        $\notinvest$ & $0~, ~0~, ~\theta-1 $ & $0~, ~0~, ~0$
        \end{game}
        };
    \end{tikzpicture}
    \caption{A two-follower example, with the leader's payoff listed first, the row follower's second, and the column follower's third.}
    \label{fig:example}
\end{figure}

%% file: equilibrium.tex
\section{Analysis and Main Results} \label{sect_signaling_game}

In this section, we present the main results of our analysis. First, we define our solution concept, $\Delta$-\textit{rationalizability}, and proceed to derive the sets of rationalizable type-strategy profiles for the leader and the followers. Then, we identify a necessary 
and sufficient condition under which the model exhibits unique rationalizable behavior. Finally, we discuss our results and the analysis's implications for efficiency.

\medskip

\subsection{Rationalizable Behavior}
Our solution concept is $\Delta$-rationalizability of \cite{battigalli_siniscalchi_2003}, which extends \citeapos{pearce_1984} notion of extensive-form rationalizability to games with incomplete information.
The ``$\Delta$'' in $\Delta$-rationalizability indicates a specific set of restrictions on beliefs that are required to be satisfied at each round of the iterative procedure. In our case, it is the signal structure commonly known to all players. We will show that in general, the set of action-type pairs that are $\Delta$-rationalizable constitute an interval both for the leader and the followers unless the followers' information is sufficiently noisy. This result hinges on the possibility of ``knowledge traps'': more precise information on the followers' part induces multiplicity of  $\Delta$-rationalizable type-strategy profiles which can lead to serious inefficiencies.


Before providing a formal definition of the procedure, we introduce the following notation. Recall that $\Action_L=\{\invest,\notinvest\}$ is the action set of the leader. To simplify notation, we also consider it to be the set of possible (non-terminal) histories of the game. We, therefore, let $a_L \in \Action_L$ denote the action chosen by the leader and $h \in \Action_L$ the corresponding history.
A strategy for follower $j \in F$ is a mapping, $s_j: \Action_L \to \Action_j$, that maps history $h$ into action $s_j(h)$.
Let $S_j$ be the set of strategies for follower $j$. The sets of all possible types of the leader and follower $j$ are $\Theta$ and $X_j$, respectively. We call $(\theta_L , a_L )\in \Theta \times \Action_L$ a type-strategy pair for the leader. Likewise, $(x_j , s_j)\in X_j \times S_j$ is a type-strategy pair for follower $j$.
Players' interim beliefs are conditional probabilities, derived from the Bayes' rule, about the type-strategy pairs of their opponents. Specifically, an interim belief of leader $\theta$ is $\mu_L (\cdot \vl \theta) \in \Delta \left(X \times S \right)$, 
where $X \times S = \prod_{j\in F} X_j \times S_j$
with a generic element $(x, s) = (x_j, s_j)_{j \in F}$, 
and the interim belief for follower $j$ given type $x_j$ and history $h$ is $\mu_j (\cdot \vl x_j,h) \in \Delta \left(\Theta \times \prod_{k \neq j} (X_{k}\times S_{k})\right)$.

For leader $\theta$, exerting effort, $a_L = \invest$, is the best response with respect to a belief $\mu_L(x, s \vl \theta)$
if
\[
\int_{(x, s)} u(\theta, A_{-L}(s))\dd \mu_L(x, s \vl \theta) > 0,
\]
where $A_{-L}(s(\invest)) = \sum_{j \in F} \one\left(s_j(\invest) = \invest \right)$.\footnote{~We assume, without loss of generality, that players break the tie by choosing not to exert effort.}
Similarly, for type $x_j$ of follower $j$ under history $h$, action $s_j(h) = \invest$ is the best response to a belief $\mu_j(\theta, x_{-j}, s_{-j} \vl x_j, h)$ if
\[
\int_{(\theta, x_{-j}, s_{-j})} u(\theta, A_{-j}(h, s_{-j}))\dd \mu_j(\theta, x_{-j}, s_{-j} \vl x_j, h) > 0,
\]
where $A_{-j}(h, s_{-j}(h)) = \chi_\invest  + \sum_{k \neq j,~k \in F} \one\left(s_k(h) = \invest \right)$ and $\chi_\invest = \one\left(h = \invest \right)$.\footnote{~Likewise, $\chi_\notinvest = \one(h = \notinvest)$.} The notion of $\Delta$-rationalizability is defined as follows.

\begin{definition}[$\Delta$-rationalizability]
Consider the following procedure.\\
(Round 0) Let $R_L^0 = \Theta \times \Action_L$ and $R_{F, \,j}^0 = X_j \times S_j$ for each $j \in F$. \\
(Round $k \geq 1$) Let $R_F^m = \prod_{j \in F} R_{F, \,j}^m$ and $R_{F, -j}^m = \prod_{\ell \neq j} R_{F, \ell}^m, m \in \{0\} \cup \NN$. Then
\begin{enumerate}
    \item[(i)] $(\theta, a_L) \in R_L^k$ if and only if $(\theta, a_L) \in R_L^{k-1}$ and there exists a belief $\mu_L(\cdot \vl \theta) \in \Delta(R_F^0)$ such that $\mu_L(R_F^{k-1} \vl \theta) = 1$ and $a_L$ is a best response with respect to $\mu_L(\cdot \vl \theta)$.
    
    \item[(ii)] For every follower $j \in F$, $(x_j, s_j) \in R_j^{k-1}$ if and only if $(x_j, s_j) \in R_j^{k-1}$ and for each history $h$ there exists a belief $\mu_j(\cdot \vl x_j, h) \in \Delta(R_L^0 \times R_{F, \,-j}^0)$ such that $\mu_j(R_L^k \times R_{F, \,-j}^{k-1} \vl x_j, h) = 1$ and $s_j(h)$ is a best response with respect to $\mu_j(\cdot \vl x_j, h)$.
\end{enumerate}
Finally, let $R_L^\infty = \bigcap_{k=0}^\infty R_L^k$ and $R_{F, \,j}^\infty = \bigcap_{k=0}^\infty R_{F, \,j}^k$. Then an action $a_L$ is $\Delta$-rationalizable for type $\theta$ of the leader if $(\theta, a_L) \in R_L^\infty$. Analogously, a strategy $s_j$ is $\Delta$-rationalizable for type $x_j$ of follower $j$ if $(x_j, s_j) \in R_{F, \,j}^\infty$.
\end{definition}

\subsubsection*{Follower Problem}
Consider type $x$ of follower $j \in F$. 
Suppose that he believes that the leader uses a monotone strategy with threshold $z \in \RR$; that is, $a_L = \invest$ for all $\theta > z$.
Therefore, type $x$'s interim belief about $\theta$ has a truncated Gaussian distribution with density
\begin{equation} \label{interim_belief}
    \psi^h(\theta; x, z) = \begin{cases}
    \frac{\frac{1}{\sigma_F}\phi\left(\frac{\theta - x}{\sigma_F}\right)}{1 - \Phi\left(\frac{z - x}{\sigma_F}\right)}\one(\theta > z) & \text{if $h = \invest$} \\
    ~ & ~ \\
    \frac{\frac{1}{\sigma_F}\phi\left(\frac{\theta - x}{\sigma_F}\right)}{\Phi\left(\frac{z - x}{\sigma_F}\right)}\one(\theta \leq z) & \text{if $h = \notinvest$}
    \end{cases},
\end{equation}
where $\phi(\cdot)$ and $\Phi(\cdot)$ denote the standard Gaussian density function (PDF) and cumulative distribution function (CDF), respectively.
Let $\Psi^h(\cdot; \,x, z)$ be the corresponding CDF under history $h$. 
We denote by $\lambda(x) = \phi(x)/\Phi(x)$ the \textit{reversed hazard rate}. Then 
type $x$'s expectation of $\theta$ can be written as
\begin{equation} \label{eq_expectation_of_theta}
    \EE_{\theta \sim \Psi^h(\cdot; \,x, z)}[\theta] = \begin{cases}
   x + \sigma_F \lambda \left(\frac{x - z}{\sigma_F}\right)  & \text{if $h = \invest$} \\
   ~ & ~ \\
   x - \sigma_F \lambda \left( \frac{z - x}{\sigma_F} \right) & \text{if $h = \notinvest$}
   \end{cases},
\end{equation}
which has the following properties.

\begin{lemma} \label{lemma_truncated_expectation}
The interim expectations $\EE_{\theta \sim \Psi^h(\cdot; \,x, z)}[\theta]$ are strictly increasing in $x$ and $z$. Moreover,
\begin{equation*}
    \lim_{x \to -\infty} \EE_{\theta \sim \Psi^h(\cdot; \,x, z)}[\theta] = \begin{cases}
    z & \text{if $h = \invest$} \\
    - \infty & \text{if $h = \notinvest$}
    \end{cases},
\end{equation*}
and
\begin{equation*}
    \lim_{x \to \infty} \EE_{\theta \sim \Psi^h(\cdot ; \,x, z)}[\theta] = \begin{cases}
    \infty & \text{if $h = \invest$} \\
    z & \text{if $h = \notinvest$}
    \end{cases}.
\end{equation*}
\end{lemma}

Now suppose further that follower $j$ believes that other followers are using monotone strategies with threshold $x_h$ under history $h$; that is, for any 
follower $\ell \neq j$ with type $x$, 
$s_\ell(h) = \invest$ if and only if $x > x_h$. This implies that, at a given state $\theta$, the probability that follower $j$ assigns to
$k$ other followers investing equals
$\left[1 - \Phi\left((x_h - \theta)/\sigma_F\right)\right]^k$, $k \in \{0, 1, \dots, n-1\}$. Therefore
follower $j$'s expected proportion of other players investing at state $\theta$ is
\begin{align*}
    A_{-j}(\theta) & = \sum_{k=0}^{n-1}  \binom{n-1}{k} \frac{k }{n} \left[1 - \Phi\left(\frac{x_h - \theta}{\sigma_F}\right)\right]^k \Phi\left(\frac{x_h - \theta}{\sigma_F}\right)^{n-1-k} +\frac{\chi_{\invest}}{n}\\
    & = \frac{n-1}{n}\left[\left(1 - \Phi\left(\frac{x_h - \theta}{\sigma_F}\right)\right) \right] + \frac{\chi_\invest}{n}.
\end{align*}
The second equality follows from
the binomial identity $\sum_{k=0}^{n-1} \binom{n-1}{k} k (1-q)^k q^{n-1-k} = (n-1)(1-q)$. Since the leader's action is observable, follower $j$ has certainty about receiving the network benefit $\chi_\invest/n$.
Thus, we may write the payoff to choosing $a_j =\invest$ for type $x$, under history $h$, as
\begin{equation} \label{follower_payoff}
   \pi_F^h(x; z, x_h) = \EE_{\theta \sim \Psi^h(\cdot; \,x, z)} \left[ \theta - \frac{n-1}{n} \Phi\left(\frac{x_h -\theta}{\sigma_F}\right)   \right] - \frac{\chi_\notinvest}{n}.
\end{equation}
We then have the next lemma.

\begin{lemma} \label{lemma_x_payoff}
The follower payoffs $\pi_F^h(x; z, x_h)$ are strictly increasing in $x$ and $z$, and are strictly decreasing in $x_h$. Moreover,
\begin{equation*}
    \lim_{x \to -\infty} \pi_F^h(x; z, x_h) = \begin{cases}
z - \frac{n-1}{n}\Phi\left( \frac{x_\invest - z}{\sigma_F} \right) & \text{if $h = \invest$} \\
-\infty & \text{if $h = \notinvest$}
\end{cases}
\end{equation*}
and
\begin{equation*}
    \lim_{x \to \infty} \pi_F^h(x; z, x_h) = \begin{cases}
\infty & \text{if $h = \invest$} \\
z - \frac{1}{n} + \frac{n-1}{n}\Phi\left( \frac{x_\notinvest - z}{\sigma_F} \right) & \text{if $h = \notinvest$}
\end{cases}.
\end{equation*}
\end{lemma}

\medskip
\subsubsection*{Leader Problem}
Suppose that, under history $h = \invest$, followers use monotone strategies with threshold $x_\invest \in \RR$; that is, $s_j(\invest) = \invest$ for $x_j > x_\invest$.\footnote{~Since followers are \emph{ex ante} identical, assuming a common threshold is without loss.} If the leader chooses $a_L=\invest$, then the expected aggregate action is given by
\begin{equation*}
    A_{-L}(\theta) = 1 - \Phi\left(\frac{x_\invest - \theta}{\sigma_F}\right).
\end{equation*}
Note that the behavior of followers matters to the leader only when $a_L = \invest$; otherwise, she obtains a payoff of zero by taking the safe action $a_L = \notinvest$. The payoff to choosing $a_L=\invest$ for type $\theta$ is therefore
\begin{equation} \label{leader_payoff}
    \pi_L (\theta; x_\invest)= \theta - \Phi\left(\frac{x_\invest - \theta}{\sigma_F}\right).
\end{equation}
It is immediate to see that $\pi_L(\theta; x_\invest)$ is strictly increasing in $\theta$ and crosses zero only once from below. Thus, the leader's best response to $x_\invest$ is the unique solution to $\pi_L(\theta; x_\invest) = 0$.

\subsection{Main Results} \label{main_results}

We first provide an intuitive explanation of how $\Delta$-rationalizability proceeds. Before the procedure starts, all players deem all type-strategy pairs possible. Let $\thetalow^0 = \xhlow^ 0 = -\infty$ and $\thetaup^0 = \xhup^0 = \infty$ for each history $h$. We call the former \textit{lower dominance bounds} and the latter \textit{upper dominance bounds}.
Note that the payoff to the leader, given $\xilow^0$ and $\xiup^0$, satisfies the standard two-sided ``limit dominance'' property of global games \citep{morris_shin_2003}, with the \textit{dominance regions} being $(-\infty, 0)$ and $(1, \infty)$. That is, exerting no effort ($a_L = \notinvest$) is dominant for all types $\theta < 0$, and exerting effort ($a_L = \invest$) is dominant for all types $\theta > 1$.
This implies that
the leader will eliminate, in Round 1, all type-action pairs $(\theta, \invest)$ with $\theta < \thetalow^1 = 0$ and $(\theta, \notinvest)$ with $\theta > \thetaup^1 = 1$. 

By knowing the leader's dominance bounds $\thetalow^1$ and $\thetaup^1$, each follower can infer from $h = \invest$ that this decision cannot be made by a type $\theta < \thetalow^1$. 
This, in turn, 
determines each follower's dominance regions. For type $x$ of a follower, Lemma \ref{lemma_x_payoff} implies that
the worst-case payoff equals 
\begin{equation*}
   \pi_F^\invest(x; \thetalow^1, \xiup^0) =   \EE_{\theta \sim \Psi^\invest(\cdot; \,x, \thetalow^1)} [\theta] - \frac{n-1}{n}.
\end{equation*}
Let $\xiup^1$ be the unique solution to $\EE_{\theta \sim \Psi^\invest(\cdot; \,\xiup^1, 0)} [\theta] = \frac{n-1}{n}$. Exerting no effort is never a best response for $x > \xiup^1$ because the worst-case payoff is strictly increasing in $x$. But since the best-case payoff is positive for all $x$:
\begin{equation*}
   \pi_F^\invest(x; \thetaup^1, \xilow ^0) = \EE_{\theta \sim \Psi^\invest(\cdot; \,x, \thetaup^1)}[\theta] > 0,
\end{equation*}
the subgame under $h = \invest$ violates 
the two-sided limit dominance property because
exerting effort is not strictly dominated for any type $x$.\footnote{~See \cite{baliga_sjostrom_2004} and \cite{bueno_de_mesquita_2010} for applications with one-sided limit dominance but different signal structures.} 
This implies that the lower dominance bound yields $\xilow^1 = -\infty$. 

Now, under history $h = \notinvest$, followers know that it must be leader $\theta \leq \thetaup^1$ that has chosen not to exert effort. The subgame exhibits no upper dominance region because the worst-case payoff to any follower type $x$
\begin{equation*}
    \pi_F^\notinvest(x; \thetalow^1, \xiup ^0) = \EE_{\theta \sim \Psi^\notinvest(\cdot; \,x, \thetalow^1)}[\theta] - 1 < 0
\end{equation*}
is negative. Thus, $\xnup^1 = \infty$. The lower dominance bound is given by the unique solution $\xnlow^1$ to 
\begin{equation*}
    \pi_F^\notinvest(\xnlow^1; \thetaup^1, \xnlow^0) = \EE_{\theta \sim \Psi^\notinvest(\cdot; \,\xnlow^1, \thetaup^1)}[\theta] - \frac{1}{n} = 0. 
\end{equation*}
In sum, each follower $j$ will delete type-strategy pairs $(x, s_j)$ such that (i) $x > \xiup^1$ and $s_j(\invest) = \notinvest$, and (ii) $x < \xnlow^1$ and $s_j(\notinvest) = \invest$.

In Round 2, $\thetalow^2$, $\xilow^2$, and $\xnup^2$ are given analogously.  The leader's upper dominance bound, $\thetaup^2$, is the unique solution to 
\begin{equation*}
    \pi_L(\thetaup^2; \xiup^1) = \thetaup^2 - \Phi\bigg( \frac{ \xiup^1 - \thetaup^2 }{\sigma_F} \bigg) = 0.
\end{equation*}
Moreover, Lemma \ref{lemma_x_payoff} implies that followers' upper dominance bound under history $h = \invest$ is the unique value of $\xiup^2$ that solves 
\begin{equation*}
    \pi_F^\invest(\xiup^2; \thetalow^2, \xiup^1) = \EE_{\theta \sim \Psi^\invest(\cdot; \, \xiup^2, \thetalow^2)} \left[\theta - \frac{n-1}{n}\Phi\left(\frac{\xiup^1 - \xiup^2}{\sigma_F}\right)\right ] = 0,
\end{equation*}
and the lower dominance bound under history $h = \notinvest$ is given by the unique solution $\xnlow^2$ to
\begin{equation*}
        \pi_F^\notinvest(\xnlow^2; \thetaup^2, \xnlow^1) = \EE_{\theta \sim \Psi^\invest(\cdot; \, \xnlow^2, \thetaup^2)} \left[\theta - \frac{n-1}{n}\Phi\left(\frac{\xnlow^1 - \xnlow^2}{\sigma_F}\right)\right ] - \frac{1}{n} = 0.
\end{equation*}

A similar argument goes for all Rounds $k > 2$.
The iteration procedure ultimately yields six sequences. We summarize their properties in the following lemma.

\begin{lemma} \label{lemma_rat_seq}
The sequences are such that:\\
(a) $(\thetalow^k)_{k=0}^\infty$ is such that $\thetalow^k = \thetalow = 0$ for all $k \geq 1$; \\
(b) $(\thetaup^k)_{k=0}^\infty$ is strictly decreasing and bounded below; \\
(c) $(\xilow^k)_{k=0}^\infty$ is such that $\xilow^k = \xilow = -\infty$ for all $k \geq 0$; \\
(d) $(\xiup^k)_{k=0}^\infty$ is strictly decreasing; \\
(e) $(\xnlow^k)_{k=0}^\infty$ is strictly increasing; \\
(f) $(\xnup^k)_{k=0}^\infty$ is such that $\xnup^k = \xnup = \infty$ for all $k \geq 0$.
\end{lemma}

By the monotone convergence theorem, $\thetaup^k$ converges to $\thetaup$ as $k \to \infty$. Moreover, $\thetaup$ is the unique solution to
\begin{equation} \label{leader_upper}
    \pi_L(\thetaup; \xiup) = 0,
\end{equation}
where $\xiup = \lim_{k \to \infty} \xiup^k$, and hence $\thetaup < 1$.
If $\xiup > - \infty$, it solves 
\begin{equation} \label{follower_invest_upper}
    \pi_F^\invest(\xiup; \thetalow, \xiup) = 0;
\end{equation}
otherwise $\xiup = -\infty$. Similarly, let $\xnlow = \lim_{k \to \infty}\xnlow^k$, and $\xnlow$ solves
\begin{equation} \label{follower_notinvest_lower}
        \pi_F^\notinvest(\xnlow; \thetaup, \xnlow) = 0
\end{equation}
if a solution exists. Otherwise $\xnlow^k$ diverges to $\xnlow = \infty$.
We now state the main result of the paper.

\begin{proposition} \label{prop_unique_rat}
The $\Delta$-rationalizable sets are $R_L^\infty = R_L^0 \setminus \overline{R}_L^\infty$ and $R_{F, \,j}^\infty = R_{F, \,j}^0 \setminus \overline{R}_{F, \,j}^\infty$, where
\begin{equation*}
     \overline{R}_L^{\infty} = \left\{(\theta, a_L) \vl~ a_L = \invest ~\text{if}~ \theta \leq 0 ~\text{and}~ a_L = \notinvest ~\text{if}~ \theta > \thetaup \right\}
\end{equation*}
and
\begin{equation*}
    \overline{R}_{F, \,j}^\infty = \left\{(x_j, s_j) \vl ~ s_j(\invest) = \notinvest ~\text{if}~x_j > \xiup ~\text{and}~s_j(\notinvest) = \invest ~\text{if}~ x < \xnlow \right\}.
\end{equation*}
Moreover, there exists a unique $\widehat{\sigma}_F$ such that there is a unique $\Delta$-rationalizable strategy profile with 
$(\thetaup, \xiup, \xnlow) = (0, -\infty, \infty)$ if and only if $\sigma_F > \widehat{\sigma}_F$.

\end{proposition}

In words, Proposition \ref{prop_unique_rat} conveys the message that if the leader has a sufficient informational advantage (i.e.,
$\sigma_F > \widehat{\sigma}_F$), then  
the unique $\Delta$-rationalizable strategy profile features leader type $\theta$ choosing $a_L = \invest$ when $\theta > 0$ and $a_L = \notinvest$ otherwise, and all follower types imitating the leader's action. This leads to a fully efficient outcome. However, 
when followers have relatively precise information (i.e.,
$\sigma_F \leq \widehat{\sigma}_F$), both actions become rationalizable for leader types $\theta \in (0,\thetaup]$ and for follower types in $(-\infty,\xiup]$ if the leader chooses to exert effort and in $(\xnlow, \infty)$ if the leader chooses to exert no effort. Thus, the leader does not necessarily choose the efficient action ($a_L = \invest$) when $\theta \in (0, \thetaup]$ for fear that followers might coordinate against her and choose the inefficient action.

Moreover, whenever we obtain unique rationalizable behavior under history $h=\invest$, we also do so under history $h=\notinvest$. This is because followers understand that the state is negative which implies that $\pi_F^{\notinvest} (\xnlow;\thetaup ,\xnlow)=0$ has no solution. If, however, we get multiplicity of rationalizable profiles under history $h=\invest$, we may or may not get multiplicity under history $h=\notinvest$. This will depend on whether $\pi_F^{\notinvest} (\xnlow;\thetaup ,\xnlow)=0$ has solutions or not for the particular value of $\sigma_F < \widehat{\sigma}_F$ considered. It should also be noted that when the necessary and sufficient condition is not satisfied, then the values of $(\thetaup, \xiup, \xnlow)$ depend on the value of $\sigma_F$ and, thus, the game features noise-dependent selection. We plot the values of $\thetaup$ in Figure 2. Figure \ref{fig_n_vs_sigma_f_hat} illustrates how the value of $\widehat{\sigma}_F$ changes with the number of followers.

\begin{figure}[htbp]
    \centering
    \includegraphics[scale=0.85]{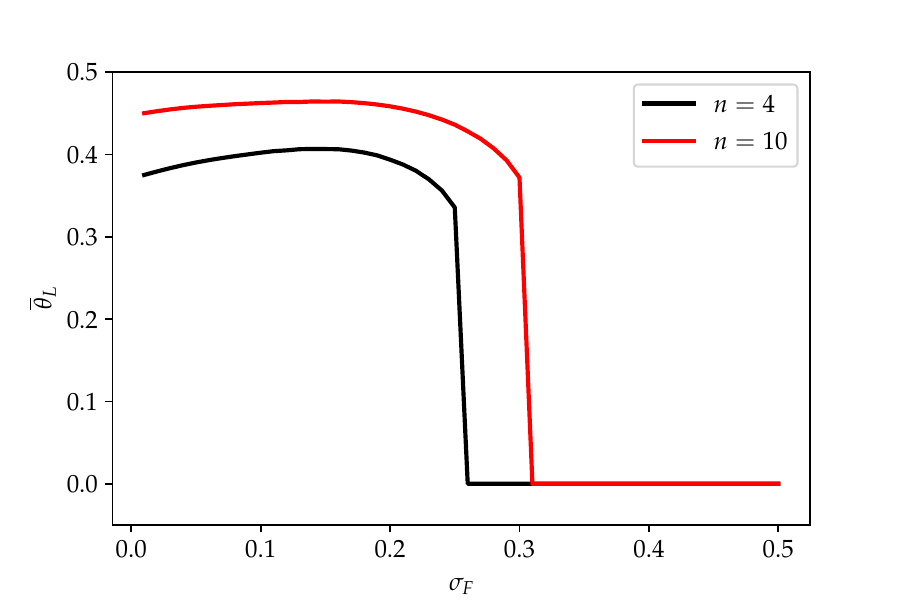}
    \caption{The value of $\thetaup$ as $\sigma_F$ changes.}
    \label{theta_up_vs_sigma_f_hat}
\end{figure}

The next result shows what happens in the limit as $\sigma_F$ tends to zero.  
\begin{proposition} \label{prop_limit}
    In the limit as $\sigma_F \to 0$, the dominance bounds $\thetaup \to (n-1)/(2n)$, $\xiup \to (n-1)/(2n)$, and $\xnlow \to \infty$.
\end{proposition}

It is worth noting that the leader's and followers' upper threshold is given by $(n-1)/2n$ which corresponds to the ``risk dominant'' strategy profile of the subgame given the spillover benefit of the leader's action. This is because, in this extreme, the followers are allowed to have beliefs that completely shut down the informational role of the leader. In particular, this would be the unique rationalizable behavior of an alternative game, with the same payoffs as in this subgame, where there is no leader and the beliefs about the state are given by the Bayesian updating of the prior after followers receive their signals. This will be better illustrated in the next section where we interpret our results in terms of "conditional rank beliefs".

\begin{figure}[htbp]
    \centering
    \includegraphics[scale=0.8]{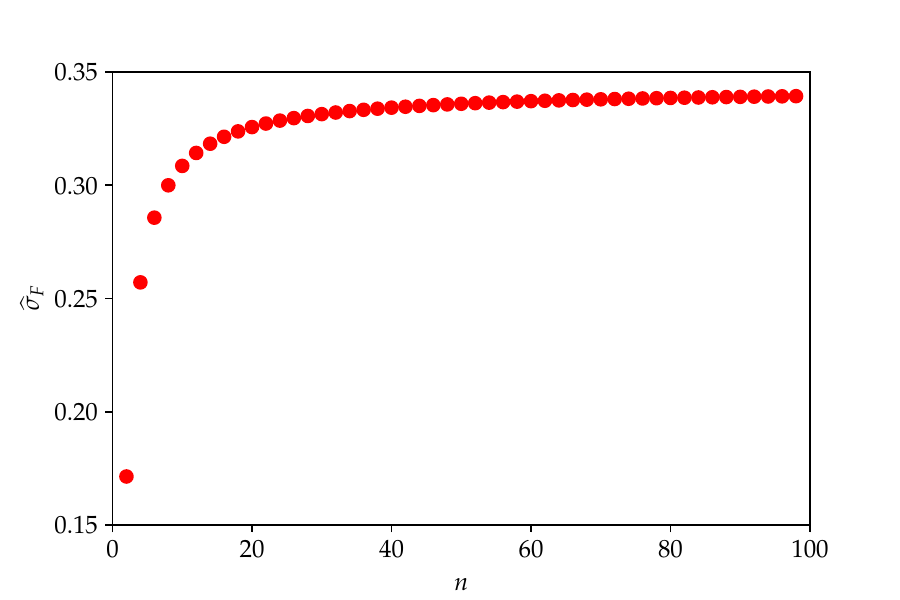}
    \caption{ The value of $\widehat{\sigma}_F$ increases with $n$.}
    \label{fig_n_vs_sigma_f_hat}
\end{figure}

\subsection{The Signaling Role of Leader and Information Traps}

It is now evident that when the condition $\sigma_F > \widehat{\sigma}_F$ is met, we obtain \emph{efficient leadership}: the leader chooses action $a_L = \invest$ whenever it socially desirable to do so (i.e., $\theta > 0$) and chooses $a_L = \notinvest$ otherwise. In this case, the outcome of the game corresponds to the fully efficient subgame perfect equilibrium of the complete information game, which may not come as a surprise.
Indeed, \cite{komai_et_al_2007} reach a similar conclusion in a different framework, which focuses on the signaling role of leaders primarily in organizational economics settings. However, our global games perspective speaks to leadership in a variety of scenarios, such as regime change, bank runs, and currency attacks. Additionally, our result reinforces theirs by deriving it using a weaker solution concept.

However, the significance of our contribution lies in the converse direction, demonstrating that if followers possess sufficiently precise private information, the leader may choose a socially undesirable action. The leader may fear that her well-informed followers might make the ``wrong'' decision, which could, in turn, compel her to act inefficiently by selecting $a_L = \notinvest$ even when $\theta > 0$.
Furthermore, even if $\theta < 0 $ and the leader chooses $\notinvest$, well-informed followers might be tempted to choose $\invest$, which is socially undesirable in this scenario because they do not know what type of leader chose action $\notinvest$. Therefore, we may encounter an \textit{information trap}: better-informed followers might make the leader choose incorrectly, or even if the leader does choose the efficient action, they, themselves, might not do so. On the other hand, if these followers were ``kept in the dark,'' efficient coordination would be achieved as the unique rationalizable outcome of the game. This is surprising in the sense that, we do not obtain ``limit uniqueness'' but rather ``limit multiplicity'' of rationalizable profiles, contrary to the standard results in the global games literature (see, for example, \cite{frankel_et_al_2003}).

Therefore, it is clear that the role of the leader is undermined by the more precise information the followers hold. Our model uncovers two opposite forces that compete with each other: we term them the signaling effect and the miscoordination effect. The necessary and sufficient condition derived in Proposition \ref{prop_unique_rat} makes certain that the signaling effect dominates the miscoordination effect and, as a result, ensures that the efficient outcome is realized. We now proceed to shed more light on these two forces and make the tension uncovered clearer.

\subsubsection{Signaling Effect versus Miscoordination Effect: Why Multiplicity Happens?}

To see why multiplicity presents itself, we will analyze the subgame after history $h$ and consider the rationalizable profiles of followers' type-strategy pairs. Notice that given the action of the leader, the subgame appears to be a global game except we have one-sided dominance regions. To make the intuition clearer it is useful to consider \textit{monotone strategies}. In that regard, consider type $x$ of follower $j$.  This type of follower $j$ does not know which threshold the leader used to make the choice that led to history $h$ being realized, so let $z$ denote this threshold.  Now, define follower $j$'s \textit{conditional rank belief} as the probability he assigns to the event that the other follower's type $x_k$ is at most his own ($x_j = x$) conditional on history $h$. We have that:

\begin{equation} \label{rank_belief}
    R^h(x; z)= \prob(x_k \leq x_j \vl x_j = x, h) = \frac{1}{2}\left[\Phi\left(\frac{x - z}{\sigma_F}\right) + \chi_\notinvest \right] 
\end{equation}
This definition is a direct extension of the rank belief function introduced in \cite{morris_et_al_2016} and \cite{morris_yildiz_2019}. 

Assume that follower $j$ conjectures that his opponents in the subgame will use a threshold $x_h$. Then, the expected payoff to choosing $\invest$ under history $h$ is given by:
\begin{equation} \label{follower_payoff_rank_belief}
    \pi_F^h(x; z, x_h) = \EE_{\theta \sim \Psi^h(\cdot; \,x, z)} \left[ \theta - \frac{n-1}{n} \Phi\left(\frac{x_h-\theta}{\sigma_F}\right)   \right] - \frac{\chi_\notinvest}{n}.
\end{equation}
Now, consider the type of follower $j$ whose signal is equal to the conjectured threshold of followers $-j$. Then, we can write this type's expected payoff to choosing $\invest$ as 
\begin{equation*}
    \pi_F^h(x_h; z, x_h) = \EE_{\theta \sim \Psi^h(x_h; \,z, x_h)} \left[ \theta \right] + \frac{n-1}{n}\left[1 - R^h(x_h; z)\right] + \frac{\chi_\invest}{n} - 1.
\end{equation*}

In order for type $x_h$ of follower $j$ to be indifferent between choosing $\invest$ and $\notinvest$ it must be the case that $x_h$ must solve for all $h$
\begin{equation} \label{eq_eqm_follower}
   \EE_{\underbrace{\theta \sim \Psi^h(\cdot; \,z, x_h)}_{\text{signaling effect}}} \left[ \theta \right] -  \underbrace{\frac{\chi_\notinvest}{n}}_{\text{spillover from leader's action}}= \underbrace{\frac{n-1}{n}R^h(x_h; z)}_{\text{miscoordination effect}} 
\end{equation}
The left-hand side of this equation captures the twofold effect of the leader's choice. First, the signaling effect simply says that by observing history $\invest$, it must be the case that $\theta>z$. The rationality of the leader implies that $z\geq 0$ and this is common knowledge among followers. Second, if the leader chose $\notinvest$ there is a negative spillover benefit to followers. The right-hand side captures the miscoordination effect: given the leader's threshold, if a follower $-j$ invests only if his type is greater than $x_h$, then type $x_h$ of follower $j$ faces an expected loss equal to (1/n) times the probability of this event, which is given exactly by the conditional rank belief function. 
Notice that if type $x_h$ of follower $j$ is indifferent between $\invest$ and $\notinvest$, then any type lower than $x_h$ will choose $\notinvest$ given the conjectured strategies.

Focus on history $h=\invest$. An analogous argument holds for history $h=\notinvest$. Notice that we can write Equation (\ref{eq_eqm_follower}) under $h = \invest$ as 
\begin{equation} \label{eq_eqm_follower_rewrite}
    x_\invest + \sigma_F \lambda\left(\frac{x_\invest - z}{\sigma_F}\right) = \frac{n-1}{n}\left[\frac{1}{2}\Phi\left(\frac{x_\invest - z}{\sigma_F}\right)\right]
\end{equation}
 Since the leader choosing action $\invest$ whenever $\theta>0$ and followers imitating the leader's action is always rationalizable behavior, to get uniqueness,
Equation (\ref{eq_eqm_follower_rewrite}) must either have no solution or have exactly one solution for $z=\thetalow=0$. In the former case we obtain the efficient outcome, while in the latter, we obtain uniqueness but full efficiency is not achieved.
Observe that
the noise, $\sigma_F$, affects both the signaling and the miscoordination effect. Moreover, the former is increasing in $\sigma_F$. In contrast, the latter does not change monotonically with $\sigma_F$ but features rapid slope change around $x_h=z$.\footnote{~It changes monotonically for $x_h<z$ and $x_h>z$ but not for all $x$.} For large values of $\sigma_F$  (i.e., $\sigma_F \geq \widehat{\sigma}_F$), the signaling effect dominates the miss-coordination effect, implying that Equation (\ref{eq_eqm_follower_rewrite}) has no solution, that is, the expected loss is always smaller than the expected return to choosing $\invest$. This means that $\xiup^k \to \xilow$ as $k \to \infty$ and hence $\invest$ is the unique rationalizable action in the subgame following $h = \invest$. See Figure \ref{fig_invest_unique} for an illustration.
As a consequence, $\thetalow = \thetaup$, $\xnlow = \xnup$ and we obtain a unique $\Delta$-rationalizable strategy profile.

\begin{figure}[htbp]
    \centering
    \includegraphics[scale=0.8]{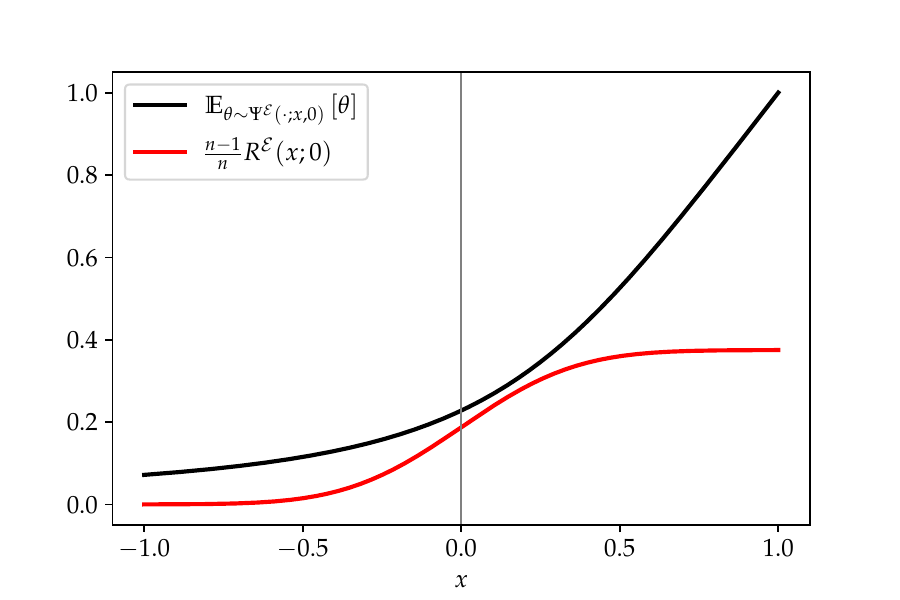}
    \caption{Unique rationalizable behavior when $\sigma_F > \widehat{\sigma}_F$ ($n=4$).}
    \label{fig_invest_unique}
\end{figure}

On the contrary, for small values of $\sigma_F$ (i.e., $\sigma_F \leq \widehat{\sigma}_F$), the signaling effect stops playing a dominant role, meaning that the expected loss may be higher than the expected return. Hence Equation (\ref{eq_eqm_follower_rewrite}) may have multiple solutions (see Figure \ref{fig_invest_multiple}). Indeed, this is the case. In particular, the largest solution corresponds to the limit $\xiup$ to which the sequence $(\xiup^k)$ will converge. This yields the fact that both actions are rationalizable for follower types $(\xilow,\xiup]$ with $\xilow=-\infty$.

At this point, the fact that $\theta>0$ has become common knowledge plays a minimal role. In fact, as $\sigma_F$ approaches zero, the signaling effect is completely obliterated and followers behave as if $\theta$ can take any value on the real line\footnote{~In fact, the monotone strategy profiles with thresholds $\xilow$ and $\xiup$ are the least and greatest Bayesian Nash equilibria that bound all rationalizable strategies in the sub-game that follows action $\invest$ of the leader \citep{van_zandt_vives_2007}.}. An interesting aspect of the model is that 
$R^\invest(\xiup, \thetalow) \to 1/2$ as $\sigma_F \to 0^+$. This guarantees that the monotone strategy profile with threshold $\xiup$ which bounds the set of rationalizable profiles in the subgame, corresponds to the ``risk dominant'' equilibrium \citep{harsanyi_selten_1988} of the subgame game given history $\invest$. That is, a follower will only choose $\invest$ if $\invest$ is a best response to a uniform belief over other followers choosing each action. However, this point cannot be supported as an equilibrium in monotone strategies of the whole game. When followers play according to the threshold $\xiup$, the leader will best respond by using a monotone strategy with threshold $\thetaup > \thetalow$, which is the limit of $\thetaup^k$ as $k \to \infty$. Thus, any leader type $\theta \in (\thetalow, \thetaup)$ will find both actions rationalizable. 

\begin{figure}[htbp]
    \centering
    \includegraphics[scale=0.8]{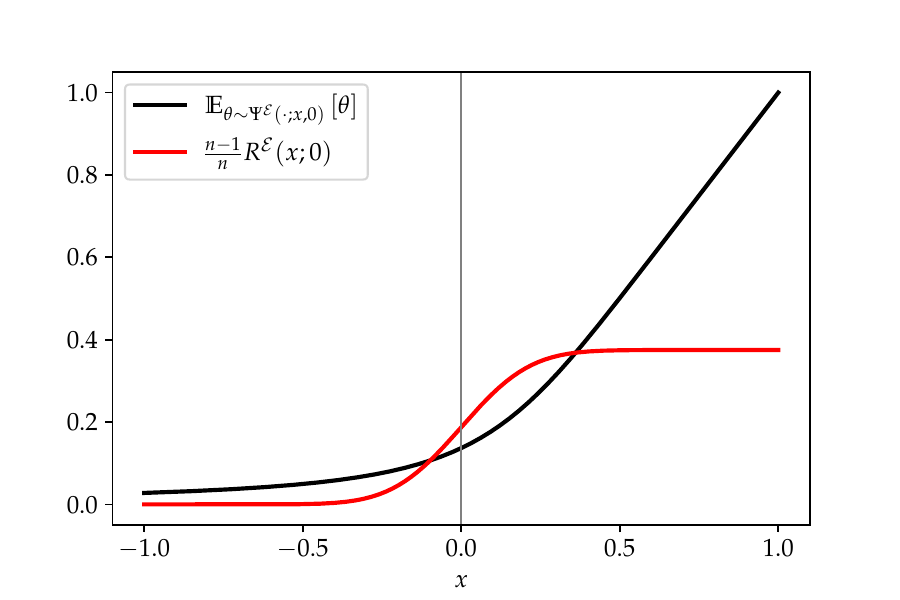}
    \caption{Multiple rationalizable type-strategy profiles when $\sigma_F \leq \widehat{\sigma}_F$ ($n=4$) with the right intersection point being the upper Bayes Nash equilibrium of the subgame.}
    \label{fig_invest_multiple}
\end{figure}

\subsection{Discussion}

\subsubsection{Equilibrium Behavior}

If the necessary and sufficient condition derived in Proposition \ref{prop_unique_rat} is satisfied, the game features unique rationalizable behavior. This immediately implies that the game features unique equilibrium behavior. In particular, the unique rationalizable strategy profile corresponds to the unique Perfect Bayesian Equilibrium of the game, which is in monotone strategies, with thresholds for the leader and followers given by $\theta_L^*=\thetaup=0 $, $\xilow=\xiup=-\infty$ and $\xnlow=\xnup=\infty$. 

As we prove in Appendix B, it is the case that when one considers monotone strategies only, the game always has a unique monotone equilibrium. However, restricting attention to these types of strategies only may be with loss, since one cannot rule out the existence of other equilibria in more complicated, non-monotone strategies. This is one of the reasons that we chose rationalizability as our solution concept.

\subsubsection{Peer-confirming Equilibrium}

Although our model features incomplete information about a fundamental state, the result is consistent with the prediction delivered by the peer-confirming equilibrium of \cite{lipnowski_sadler_2019}. In a leader-centered star network where all followers observe the leader's strategy, \cite{lipnowski_sadler_2019} argue that ``imitation'' may arise as the unique extensive-form peer-confirming equilibrium since followers do not observe any information that could contradict the leader's rationality. 
The leader, thus, has the advantage of inducing others to imitate her behavior.
In this regard, Proposition 1 deals with the interplay between strategic uncertainty and fundamental uncertainty by
identifying a necessary and sufficient condition under which the unique $\Delta$-rationalizable strategy profile that arises also corresponds to the unique peer-confirming equilibrium that highlights imitation.

\medskip

%% file: extension.tex
\section{Extension: Leader is not Perfectly Informed}
\subsection{Alternative Information Structure} \label{sect_noisy_signaling}

Suppose, now, that instead of perfectly learning the state $\theta$, the leader observes a noisy signal $x_L = \theta + \sigma_L \varepsilon_L$ with $\varepsilon_L$ being a standard Gaussian noise independent of $\theta$ and $\varepsilon_j$ for any $j \in F$. We deem this extension important for three reasons. First, from an applied perspective, it is more natural to assume that the leader observes a noisy signal about the state rather than its true value. Second, we document that multiplicity of rationalizable behavior is not a consequence of the one-sided dominance in the subgames. That is, the reason for multiplicity is not that the action of the leader when she is perfectly informed introduces ``too much'' common knowledge. Rather, the main result of the paper continues to hold with the (noisy) information structure that brings back the standard two-sided dominance. Finally, this extension will show that while it must be the case that the leader is fully informed to obtain efficiency, it is not enough for the leader to be better informed than the followers to obtain unique rationalizable behavior. The requirement that the noise in leader's information is smaller than the one in the followers' (i.e., $\sigma_L < \sigma_F$), is neither necessary nor sufficient for uniqueness.  

\subsection{Analysis and Rationalizable Behavior}

Call $x_L \in X_L = \RR$ the leader's type. A strategy for the leader is now a mapping $s_L: X_L \to \Action_L$. Let $S_L$ denote the strategy space for the leader. The actions and strategies of the followers remain the same as defined in Section \ref{sect_signaling_game}.
Let $\mu_L(\cdot\vl x_L) \in \Delta \left(\Theta \times X \times S \right)$ be leader $x_L$'s belief about the state and the type-strategy pairs of the followers, where $X \times S = \prod_{j\in F} (X_j \times S_j)$. Since there is no learning for the leader, the marginal of $\mu_L(\cdot \vl x_L)$ about $\theta$ has a Gaussian distribution with mean $x_L$ and variance $\sigma_L^2$.
Let $\mu_j(\cdot \vl x_j, h) \in \Delta \left(\Theta \times X_L \times X_{-j} \times S_{-j} \right)$ be follower $j$'s belief about 
the state and the type-strategy pairs of his opponents given type $x_j$ and history $h$, where $X_{-j} \times S_{-j} = \prod_{k \neq j \in F} (X_k \times S_k)$.

For type $x_L$ of the leader, exerting effort is the best response to a belief $\mu_L(\cdot \vl x_L)$ if
\begin{equation*}
    \int_{(\theta, x, s)} u(\theta, A_{-L}(s)) \dd \mu_L(\theta, x, s \vl x_L) > 0,
\end{equation*}
where $A_{-L}(s(\invest)) = \sum_{j \in F} \one(s_j(\invest) = \invest)$.
For a follower with type $x_j$, $s_j(h) = \invest$ is the best response to $\mu_j(\cdot \vl x_j, h)$ when
\begin{equation*}
    \int_{(\theta, x_L, x_{-j}, s_{-j})} u\left(\theta, A_{-j}(h, s_{-j}) \right) \dd \mu(\theta, x_L, x_{-j}, s_{-j} \vl x_j, h) > 0
\end{equation*}
where $A_{-j}(h, s_{-j}(h)) = \chi_\invest + \sum_{k \in F, \,k \neq j} \one(s_k(h) = \invest)$. The initial set of type-strategy pairs for the leader in the definition of $\Delta$-rationalizability is now given by $R_L^0 =X_L \times S_L$.

Suppose that a follower type $x$ believes that the leader uses a monotone strategy with threshold $z$; i.e., $a_L = \invest$ if and only if $x_L > z$. Then, 
upon observing $h = \invest$, type $x$'s interim belief has CDF
\begin{equation} \label{posterior_invest}
    G^\invest(\theta; \,x, z) = \frac{1}{ \Phi\left( \frac{x - z}{\sigma}\right) } \int_{-\infty}^{\theta} \frac{1}{\sigma_F} \phi\left(\frac{t- x}{\sigma_F}\right) \Phi\left(\frac{t - z}{\sigma_L}\right)\dd t,
\end{equation}
where $\sigma^2 = \sigma_F^2 + \sigma_L^2$.
Similarly, type $x$'s interim CDF under history $h = \notinvest$ is
\begin{equation} \label{posterior_not_invest}
    G^\notinvest(\theta; \,x, z) = \frac{1}{\Phi \left(\frac{z - x}{\sigma}\right) } \int_{-\infty}^{\theta}
    \frac{1}{\sigma_F} \phi\left(\frac{t-x}{\sigma_F}\right) \Phi\left(\frac{z - t}{\sigma_L}\right)\dd t.
\end{equation}
It is worth noting that $G^h(\cdot; \,x, z)$, unlike the interim beliefs $\Psi^h(\cdot; \,x, z)$ in the main model, has support over the entire real line. Thus, the subgames no longer feature one-sided dominance in the $\Delta$-rationalizability procedure.

If type $x$ believes further that other followers use monotone strategies with threshold $x_h$ under history $h$, then his payoff to choosing $\invest$ yields
\begin{equation*}
    \pi_F^h(x_j; z, x_h) = \EE_{\theta \sim G^h(\cdot; x, z) } \left[\theta - \frac{n-1}{n} \Phi\left(\frac{x_h - \theta}{\sigma_F}\right) \right]  -\frac{\chi_\notinvest}{n}.
\end{equation*}
We show in Appendix A that $\pi_F^h(x; z, x_h)$ is strictly increasing in $x$ and crosses zero once from below. Furthermore, it is strictly increasing in $z$ and strictly decreasing in $x_h$.
Type $x$'s 
conditional rank belief is given by
\begin{equation} \label{ext_rank_belief}
\begin{cases}
    R^\invest(x; z) = \prob(x_k \leq x_j \vl x_j = x, x_L > z) = \frac{1}{2} - \frac{T\left(\frac{x-z}{\sigma},  
 ~\alpha \right)}{\Phi\left(\frac{x - z}{\sigma}\right)} & \\
    ~ & \\
    R^\notinvest(x; z)  = \prob(x_k \leq x_j \vl x_j = x, x_L \leq z) = \frac{1}{2} + \frac{T\left(\frac{z - x}{\sigma}, ~\alpha \right)}{\Phi\left(\frac{ z - x}{\sigma}\right)} 
\end{cases},
\end{equation}
where $\alpha = \sigma_F/(2\sigma_L^2 + \sigma_F^2)^{1/2}$ and $T(y, a)$ is \textit{Owen's T-function}.\footnote{\label{owen_t}~Owen's T-function, first introduced by \cite{owen_1956}, is defined by 
\[
T(y,a) = \frac{1}{2\pi}\int_0^a \frac{\ee^{-(1+t^2)y^2/2}}{1+t^2} \dd t.
\]
It gives the probability of the event $\{X > y, ~0 < Y < a X \}$ when $X$ and $Y$ are independent standard Gaussian random variables. 
See \cite{savischenko_2014} and \cite{brychkov_savischenko_2016} for an overview of the function.} The derivation of (\ref{ext_rank_belief}) is given in Appendix A. When $x = x_h$, it can be shown that
\begin{equation} \label{ext_follower_fp_eqn}
    \pi_F^h(x_h; z, x_h) = \EE_{\theta \sim G^h(\cdot; \,x_h, z)}[\theta] - \frac{n-1}{n}R^h(x_h; z) - \frac{\chi_\notinvest}{n}.
\end{equation}

Now consider type $x_L$ of the leader. Suppose that leader $x_L$ believes that
followers use monotone strategies with threshold $x_\invest$ under history $h = \invest$. Then her payoff to choosing $\invest$ is
\begin{equation*}
    \pi_L(x_L; x_\invest) = x_L - \Phi\left(\frac{x_\invest -x_L}{\sigma}\right).
\end{equation*}
which is strictly increasing in $x_L$ and strictly decreasing in $x_\invest$.

The $\Delta$-rationalizability procedure again yields six sequences. We prove in Appendix A that $(\xlow^k)_{k=0}^\infty$, $(\xilow^k)_{k=0}^\infty$, and $(\xnlow^k)_{k=0}^\infty$ are increasing and bounded above, and $(\xup^k)_{k=0}^\infty$, $(\xiup^k)_{k=0}^\infty$, and $(\xnup^k)_{k=0}^\infty$ are decreasing and bounded below.
The monotone convergence theorem therefore guarantees that they converge to $\xlow$, $\xilow$, $\xnlow$, $\xup$, $\xiup$, and $\xnup$, respectively. In addition, the limits together solve the following system of equations:
\begin{equation} \label{ext_sys_eqs}
    \begin{cases}
        \pi_L(\xlow; \xilow) = 0 \\
        \pi_L(\xup; \xiup) = 0 \\
        \pi_F^\invest(\xilow; \xup, \xilow) = 0 \\ \pi_F^\invest(\xiup; \xlow, \xiup) = 0 \\
        \pi_F^\notinvest(\xnlow; \xup, \xnlow) = 0 \\
        \pi_F^\notinvest(\xnup; \xlow, \xnup) = 0
    \end{cases}.
\end{equation}

For a given pair of $(\sigma_L, \sigma_F)$, note that we can view it as a point on the ray from the origin with slope $\sigma_F/\sigma_L$. To understand which pair of $(\sigma_L, \sigma_F)$ induces unique $\Delta$-rationalizable behavior, we establish a sufficient condition on each fixed ray $\sigma_F = \gamma \sigma_L$, $\gamma \geq 0$, along which the slope parameter of Owen's T-function is given by $\alpha = \gamma/\sqrt{2 + \gamma^2}$.

\begin{proposition} \label{prop_unique_rat_ext}

The $\Delta$-rationalizable sets are $R_L^\infty = R_L^0 \setminus \overline{R}_L^\infty$ and $R_{F, \,j}^\infty = R_{F, \,j}^0 \setminus \overline{R}_{F, \,j}^\infty$, where
\begin{equation*}
     \overline{R}_L^{\infty} = \left\{(x_L, a_L) \vl~ a_L = \invest ~\text{if}~ x_L \leq \xlow ~\text{and}~ a_L = \notinvest ~\text{if}~ x_L > \xup \right\}
\end{equation*}
and
\begin{equation*}
\overline{R}_{F, \,j}^\infty =  \left\{(x_j, s_j) \vl
\text{$s_j(h) = \invest$ if $x_j \leq \underline{x}_h$ and  $s_j(h) = \notinvest$ if $x_j > \overline{x}_h$, \text{for all}~$h \in \Action_L$} \right\}
\end{equation*} 
Moreover, \\
(i) there exists $\widehat{\sigma}_L(\gamma)$ such that the game has unique $\Delta$-rationalizable behavior if $\sigma_L > \widehat{\sigma}_L(\gamma)$; \\
(ii) in the limit as $\sigma_L \to 0$ (while keeping $\sigma_F/\sigma_L = \gamma$ fixed), the game features multiplicity of $\Delta$-rationalizable behavior.
\end{proposition}

Proposition \ref{prop_unique_rat_ext} is an analog of Proposition \ref{prop_unique_rat}. It establishes that the game generally features multiplicity of rationalizable behavior. In particular, this is necessarily the case, given the ray $\gamma$, in the limit where both noises approach zero in a way that their ratio is always given by $\gamma$. Moreover, for each $\gamma$, one can increase the noises $(\sigma_F,\sigma_L)$ in a way that their ratio is given by $\gamma$ and the game features unique rationalizable behavior. Note, however, that now, even when unique rationalizable play is obtained, the thresholds the agents use depend on the noises $(\sigma_F,\sigma_L)$. This was not the case in our main model. There, as long as $\sigma_F >\widehat{\sigma}_F$, the unique rationalizable play was always the fully efficient one. Finally, the leader having more accurate information than the followers is neither necessary nor sufficient to obtain unique rationalizable behavior, since this can happen irrespective of whether $\gamma$ is greater, equal, or less than one. On the other hand, the leader being arbitrarily better informed than followers is necessary to obtain the efficient outcome.

\subsection{Discussion}

\subsubsection{Signaling Effect, Miscoordination Effect, and Multiplicity}

In a similar spirit to the analysis of the main model, one can analyze the subgame after history $h$ and consider the rationalizable profiles of followers' type-strategy pairs.  Let $z=x_L^* $, that is, $z$ be equal to the leader's threshold in the case where unique rationalizable behavior obtains. Assume that $\widehat{x}$ is the type of follower who is indifferent between choosing $\invest$ and $\notinvest$. One can rewrite Equation \ref{eq_eqm_follower} as
\begin{equation} \label{eq_f_ext}
    \EE_{\underbrace{{\theta \sim G^h(\widehat{x}; z, \widehat{x})}}_\text{signaling effect}} \left[ \theta \right] -  \underbrace{\frac{\chi_\notinvest}{n}}_{\text{externality from leader's action}}= \underbrace{\frac{n-1}{n}R^h(\widehat{x}; z)}_{\text{miss-coordination effect}} 
\end{equation} 

As we stated, in this case, Equation (\ref{eq_f_ext}) has at least one solution. To get the uniqueness of rationalizable play, it must be the case that the derivative of the conditional rank belief function is sufficiently bounded. This is not generally the case, since around $z$, for certain values of $\sigma_L $ and $\sigma_F$\footnote{~In particular, this is necessarily the case in the limit as $\sigma_L\rightarrow 0$ and $\sigma_F \rightarrow 0$ with $\sigma_F/\sigma_L= \gamma$.}, the rank belief function abruptly changes, which means that the expected payoff of the indifferent type of follower $j$ changes sign more than once. The condition of Proposition \ref{prop_unique_rat_ext} ensures that this rapid change is not enough to make the expected payoff of follower $j$ cross the $x$-axis multiple times. Similarly to the main model, if the sign change occurred more than once, the subgame would feature at least two Bayesian Nash equilibria that would correspond to the solutions of Equation \ref{eq_f_ext}. In this case, multiplicity of rationalizable behavior immediately obtains.  Such a case is given in Figure \ref{fig_invest_noisy_multiple}.

\begin{figure}[htbp]
    \centering
    \includegraphics[scale=0.8]{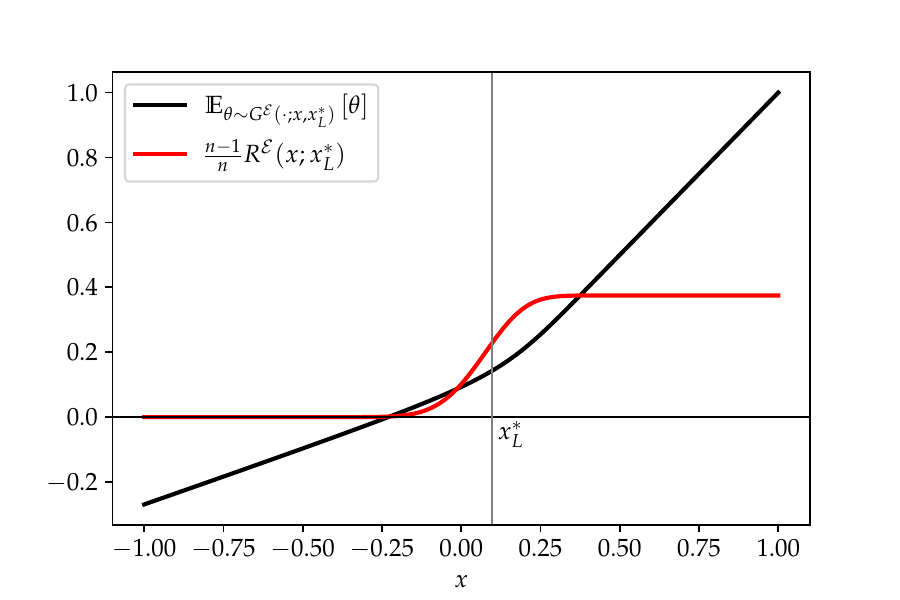}
    \caption{Multiple rationalizable type-strategy profiles.}
    \label{fig_invest_noisy_multiple}
\end{figure}

\subsubsection{Inefficiency of the Unique Outcome}
Contrary to the main model, the extension features an inefficient outcome irrespective of the uniqueness of rationalizable play. This result obtains whenever $\sigma_L$ is bounded away from zero. In the limiting case where $\sigma_L \to 0^+$ and for $\sigma_F$ sufficiently large, we recover the unique efficient $\Delta$-rationalizable profile of the main model. It is worth noting, though, that in all cases, the extensive form game leads to outcomes at least as efficient as the ones that would obtain if the game was a simultaneous move game, a prediction consistent with existing literature. This means that the presence of the leader is always helpful, even if her information is very imprecise. This is not surprising, since the leader's action apart from information carries a benefit that spills over to followers.



\subsection{A Synthesis of the Results}
One can think of the results established in Propositions 1 and 3 in the following way: In the $(\sigma_L , \sigma_F)$ space, when $\sigma_L=0$, Proposition 1 derives a necessary and sufficient condition under which the game features unique rationalizable behavior which delivers the fully efficient outcome. When $\sigma_F=0$, multiplicity immediately obtains since the followers are perfectly informed about the state. In the limit where both noises vanish and $\sigma_F/\sigma_L\to 0$,\footnote{~ This means that followers are arbitrarily better informed than the leader.} the subgame becomes a standard global game, where the leader's action carries only the positive spillover and no information. In this case, the unique rationalizable behavior features a monotone strategy for the leader and the followers with threshold $(n-1)/2n$.  When both $\sigma_L$ and $\sigma_F$ are nonzero and vanish at a rate such that their ratio is given by $\gamma$ for any $\gamma>0$, Proposition 3 establishes the multiplicity of rationalizable behavior in the limit when we move towards the origin along the fixed ray $\gamma$ and the existence of a value $\widehat{\sigma}_L(\gamma)$ such that when one is sufficiently away from the origin the game features unique rationalizable behavior. Thus, in general, the leader cannot ``discipline'' the followers on imitating her behavior and she may choose the inefficient action in the first place.

%% file: conclusion.tex
\section{So, Does One Soros Make a Difference?}

Although our model is stylized, the forces we uncover are present in a range of economic settings, in particular, financial stability and crises. In financial economics, a substantial body of literature highlights coordination failures as a key contributor to fragility. These failures occur when economic agents undertake ``destabilizing'' actions based on the anticipation that others will do the same. The outcome is a self-fulfilling crisis, commonly referred to as panic. On the other hand, it is often a reason for discussion about whether a crisis happens due to flawed fundamentals or due to panic. The global games approach has become very popular as it provides a bridge that connects these two perspectives. Using this approach, these papers (for example, \cite{corsetti_et_al_2004}) predict that the presence of a large player (our leader) may lead to a higher degree of (de)stabilization if that is the large player's preferred outcome. Our results show that this might not be necessarily the case.

\cite{corsetti_et_al_2004} analyze the question of whether large visible traders increase the vulnerability of a currency to speculative attacks. Their model features a single large investor and a continuum of small investors who decide whether to attack a currency based on their private information about the fundamentals. One of their results is that the large player's action makes small traders more aggressive compared to the case where the large player is absent.

Although our framework differs from that of \cite{corsetti_et_al_2004}, our results provide valuable insights into the question at hand. In their study, the visible trader serves as our equivalent of a leader. However, in our model, the leader is not characterized by size or power, as each follower possesses equal influence with the leader. Nonetheless, the leader's presence induces more aggressive behavior among the followers. This effect arises because our model involves a finite number of players, resulting in each individual's actions affecting the payoffs of others. Additionally, the leader's actions serve as signals to the followers, conveying a portion of her information. Even as we approach the assumption of a continuum by increasing the number of followers infinitely, this overall impact remains, driven solely by informational factors.

Our findings demonstrate that a powerful player is not necessarily needed to elicit this effect; a visible agent is sufficient. However, this holds true only if the condition that yields unique rationalizable behavior is satisfied. In \cite{corsetti_et_al_2004}, the focus is on monotone strategies to derive their results, while our approach complements theirs by examining the uniqueness or multiplicity of rationalizable behavior rather than equilibrium behavior. In cases where our sufficient condition is met, our model aligns with the predictions of \cite{corsetti_et_al_2004} (see Proposition 7 of that paper).

If that condition is not satisfied, however, then for specific pairs of ($\sigma_L,\sigma_F)$  the game features multiple rationalizable strategy profiles. This, in the language of \cite{corsetti_et_al_2004}, implies that even if the large trader is visible, the fear of follower miscoordination may deter certain types of the leader from attacking the currency. Thus, it may be the case that the attack never happens even if the fundamentals are weak, as long as the large trader cannot take down the currency on her own. Moreover, even if the large trader does attack the currency, there is still a possibility that it will survive, resulting in a negative payoff for her. This is due to the lack of guarantee that the small traders will coordinate their actions with the leader. Furthermore, even if the visible trader does not deem the attack worthwhile, the followers may still opt to do so, given their uncertainty regarding the leader's motives. Was it because the fundamentals were too strong or due to her fear of followers miscoordinating? Thus, they might still coordinate on attacking, as long as it is possible that they, on their own, can take down the currency.

A similar argument applies to models of bank runs, such as \cite{goldstein_pauzner_2005}. The role of the leader is now played by a large and visible creditor. While equilibrium beliefs would make the smaller creditors imitate the action of the large player, when non-equilibrium beliefs are taken into account, this no longer necessarily holds. Thus, even if the large creditor would ideally not run on the bank for some intermediate values of fundamentals, she may be forced to do so out of fear that smaller creditors will, even after observing her action.

In a broader sense, these models predict that as noise diminishes, the large player can attain the most favorable outcome from her perspective. This implies that a significant investor contributes to increased fragility in the context of a currency attack, whereas in a bank-run model, she fosters greater stability. However, our paper emphasizes that this prediction becomes fragile when considering rationalizable behavior. Specifically, the leader can only achieve the optimal outcome from her viewpoint when a unique rationalizable outcome obtains. Consequently, our model highlights the significance of precise information on the part of the followers, rendering the presence of the large player less influential. In fact, smaller players may coordinate their actions against her, leading to a scenario where the large player refrains from taking actions that would (de)stabilize the economy, even if the underlying fundamentals justify such actions perfectly. Of course, it should be noted that what we have just discussed is a novel effect that will be present in regime change games. Whether this effect will play a significant role, like in our setting, or not, may depend on the intricacies of the specific model considered.

\section{Why $\Delta$-rationalizability?}
Before we conclude, we discuss our choice of $\Delta$-rationalizability as our solution concept. 
The global games approach has gained significant popularity due to its ability to provide unique predictions about equilibrium behavior. However, when we deviate from the static benchmark or introduce signaling, achieving uniqueness becomes more challenging. Although one can usually establish uniqueness within the class of monotone strategies, this may exclude other equilibria involving more complex strategies which are difficult to characterize. We, therefore, choose rationalizability as our solution concept since it provides the best possible bound on equilibrium outcomes.

While the canonical rationalizability concept for static games of incomplete information, interim correlated rationalizability (ICR) (see \cite{dekel_et_al_2007}), is well established, this is not the case for multistage games. In this paper, we choose $\Delta$-rationalizability as our preferred concept because it is fully consistent with  Harsanyi’s approach and is particularly suitable for addressing the questions we investigate. Additionally, it simplifies notation and analysis by reducing the need for certain technical details. Another possible candidate is interim sequential rationalizability (ISR), a generalization of ICR to multistage games developped in \cite{penta_2012}. Another approach would involve applying ICR to the normal form of the game as described by \cite{chen_2012}. As it turns out, for our model, these three approaches yield equivalent results. This equivalence stems from the following observation: if we were to specify a type space as mandated by ICR and ISR, our game would become a game with no information in the sense of \cite{penta_2012}. Consequently, the predictions generated by ISR would be identical to those obtained by applying ICR to the normal form of the game. Given that our restrictions on first-order beliefs, $\Delta$, can be derived from the type space, Proposition 3 of \cite{battigalli_et_al_2011}  implies that these predictions coincide with the ones derived from ICR.

\section{Conclusion}
In this paper, we adopt an informational approach to the study of leadership and analyze whether a leader can coordinate her followers on a mutually desirable course of action in a framework that features strategic complementarities. Our main result shows that this is not necessarily the case. It is indeed possible for highly informed followers to undermine the leadership's role in facilitating coordination on the desired action and even cause the leader herself to choose inefficient actions. We have identified the tension between two main forces that drive our results: the signaling effect of the leader's actions and the miscoordination effect stemming from the followers' dispersed information. We have established conditions regarding the precision of information of the leader and the followers under which either effect can dominate, leading to uniqueness or multiplicity of rationalizable play and efficient or inefficient outcomes.

It is important to note that our model is highly stylized, and several questions remain open for further exploration. The most promising one, we believe, is the endogenous determination of the leader, particularly in relation to the costs associated with acquiring information. In our model, the leader is exogenously chosen, and access to information sources is assumed to be cost-free. However, if a team of individuals can acquire costly information about the state, would a leader endogenously emerge? Would players understand the potentially negative effects of everyone acquiring information and opt to have only one player, the leader, do so? We leave these and other questions for future research.

Finally, it should be noted that while our results may support the centralization of information, this conclusion only applies within the examined framework. In different scenarios, although we anticipate that precise information on the followers' part would still undermine coordination and efficient leadership as we have described, there may be other reasons why information dissemination is desired. Therefore, a broader avenue for future research is to examine under what circumstances dispersed information is preferred, despite its negative impact on facilitating coordination towards efficient actions.

%% file: appendix_a.tex
\section*{Appendix A: Proofs }

\subsection*{Proof of Lemma \ref{lemma_truncated_expectation}} First note that  type $x$'s belief about $\theta$ has a Gaussian distribution with mean $x$ and variance $\sigma_F^2$. 
Recall that $\lambda(x) = \phi(x)/\Phi(x)$ is the reversed hazard rate of a standard Gaussian random variable. 
Rewrite Equation (\ref{eq_expectation_of_theta}) as follows:
\begin{equation*} 
    \EE_{\theta \sim \Psi^h(\cdot; \,x, z)}[\theta] = \begin{cases}
   \sigma_F \left[y + \lambda(y) \right] + z  & \text{if $h = \invest$} \\
   ~ & ~ \\
   -\sigma_F \left[-y + \lambda(-y) \right] + z & \text{if $h = \notinvest$}
   \end{cases},
\end{equation*}
where $y = (x - z)/\sigma_F$. 
Since $-1 < \lambda'(x) < 0$ \citep{sampford_1953},
$x + \lambda(x)$ is strictly increasing in $x$. It is now straightforward to see that $\EE_{\theta \sim \Psi^h(\cdot; \,x, z)}[\theta]$ is strictly increasing in $x$ and $z$ because $\lambda(\cdot)$ is a  strictly decreasing function and $\sigma_F > 0$.

Using the derivative formula $\phi'(x) = -x \phi(x)$ and L'H\^{o}pital's rule, we have $\lim_{x \to -\infty} x + \lambda(x) = \lim_{x \to -\infty} \phi(x)/\phi'(x) = 0$ and hence $x + \lambda(x) > 0$ for all $x \in \RR$. We also have $\lim_{x \to \infty} x + \lambda(x) = \infty$ because $\lim_{x \to \infty} \lambda(x) = 0$. Thus,
\begin{equation*}
    \EE_{\theta \sim \Psi^\invest(\cdot; \,x, z)}[\theta] \to \begin{cases}
   z & \text{as $x \to -\infty$} \\
   \infty & \text{as $x \to \infty$}
   \end{cases},
\end{equation*}
and
\begin{equation*}
    \EE_{\theta \sim \Psi^\notinvest(\cdot; \,x, z)}[\theta] \to \begin{cases}
   -\infty & \text{as $x \to -\infty$} \\
   z & \text{as $x \to \infty$}
   \end{cases}.
\end{equation*}
This completes the proof. \qed

\medskip
\subsection*{Proof of Lemma \ref{lemma_x_payoff}}
(Part 1)
By Lemma \ref{lemma_truncated_expectation}, it suffices to show that $\Psi^h(\theta; x, z)$ is increasing in $x$ and $z$ in the sense of first-order stochastic dominance
because $\Phi\left((x_h - \theta)/\sigma_F\right)$ is strictly increasing in $\theta$ and strictly decreasing in $x_h$.
 
Given $h = \invest$. For $x' > x$, we have
\begin{equation*}
    \frac{\psi^\invest(\theta; x', z)}{\psi^\invest(\theta; x
    , z)} = \frac{\phi\left(\frac{\theta - x'}{\sigma_F}\right)}{\phi\left(\frac{\theta - x}{\sigma_F}\right)} \cdot \frac{\Phi\left(\frac{x - z}{\sigma_F}\right)}{\Phi\left(\frac{x' - z}{\sigma_F}\right)}.
\end{equation*}
Since $\phi(\cdot)$ is   log-concave, $\phi\left((\theta - x')/\sigma_F\right)/\phi\left((\theta - x)/\sigma_F\right)$ is increasing in $\theta$; i.e., $\phi(\cdot)$ is a P\'{o}lya frequency function of order 2 \citep[Proposition 2.3]{saumard_wellner_2014}, and so is $\psi^\invest(\theta; x', z)/\psi^\invest(\theta; x, z)$. This implies that $\psi^\invest(\theta; x, z)$ is log-supermodular in $(\theta, x)$, or, equivalently, $\Psi^\invest(\theta; x',z)$   dominates $\Psi^\invest(\theta; x,z)$ in the monotone likelihood ratio order. Thus, $\Psi^\invest(\theta; x',z)$   first-order stochastically dominates $\Psi^\invest(\theta; x,z)$. 

By definition, we have
\begin{equation*}
    \Psi^\invest(\theta; x, z) = 1 - \frac{\Phi\left(\frac{x - \theta}{\sigma_F}\right)}{\Phi\left(\frac{x - z}{\sigma_F}\right)}.
\end{equation*}
It is decreasing in $z$; that is, $\Psi^\invest(\theta; x, z') < \Psi^\invest(\theta; x, z)$ for $z' > z$ and for all $\theta \in (z, \infty)$, where $\Psi^\invest(\theta; x, z') = 0$ when $\theta \in (z, z']$. This proves that $\Psi(\theta; x, z')$ dominates $\Psi(\theta; x, z)$ in the first-order stochastic dominance sense. The proof for $\Psi^\notinvest(\theta; x, z)$ is analogous.

\noindent (Part 2)  Since $\Phi\left( (x_h - \theta)/\sigma_F \right)$ is bounded, it follows from Lemma \ref{lemma_truncated_expectation} that $\pi_F^\notinvest(x; z, x_\notinvest)$ $\to -\infty$ as $x \to -\infty$ and $\pi_F^\invest(x; z, x_\invest)$ $\to \infty$ as $x \to \infty$. Now under $h = \invest$,
\begin{align*}
    \EE_{\theta \sim \Psi^\invest(\cdot; \,x, z)} \left[ \Phi\left(\frac{x_\invest - \theta}{\sigma_F}\right) \right] & = \frac{1}{\Phi\left( \frac{x - z}{\sigma_F} \right)} \int_z^\infty \Phi\left( \frac{x_\invest - t}{\sigma_F} \right) \frac{1}{\sigma_F} \phi\left( \frac{x - t}{\sigma_F} \right) \dd t \\
    & = \frac{1}{\Phi\left( \frac{x - z}{\sigma_F} \right)} \int_{-\infty}^{\frac{x - z}{\sigma_F}} \phi(\eta) \Phi\left( \frac{x_\invest - x}{\sigma_F} + \eta \right) \dd \eta \\
    & \to \Phi\left( \frac{x_\invest - z}{\sigma_F} \right) \quad \text{as $x \to -\infty$}.
\end{align*}
The second equality is given by a change of variable $\eta = (x - t)/\sigma_F$, and the limiting result is a consequence of applying L'H\^{o}pital's rule. Thus, Lemma \ref{lemma_truncated_expectation} implies that $\lim_{x \to -\infty} \pi_F^\invest(x; z, x_\invest)$ $= z - ((n-1)/n)\Phi(( x_\invest - z )/\sigma_F)$.

Similarly, under $h = \notinvest$, we have 
\begin{align*}
    \EE_{\theta \sim \Psi^\notinvest(\cdot; \,x, z)} \left[ \Phi\left(\frac{x_\notinvest - \theta}{\sigma_F}\right) \right] & = \frac{1}{\Phi\left( \frac{z - x}{\sigma_F} \right)} \int_{-\infty}^z \Phi\left( \frac{x_\notinvest - t}{\sigma_F} \right) \frac{1}{\sigma_F} \phi\left( \frac{x - t}{\sigma_F} \right) \dd t \\
    & = \frac{1}{\Phi\left( \frac{z - x}{\sigma_F} \right)} \int_{\frac{x - z}{\sigma_F}}^\infty \phi(\eta) \Phi\left( \frac{x_\notinvest - x}{\sigma_F} + \eta \right) \dd \eta \\
    & \to \Phi\left( \frac{x_\notinvest - z}{\sigma_F} \right) \quad \text{as $x \to \infty$}.
\end{align*}
Thus, by Lemma \ref{lemma_truncated_expectation}, $\lim_{x \to \infty} \pi_F^\notinvest(x; z, x_\notinvest) = z - 1/n - ((n-1)/n) \Phi((x_\notinvest - z)/\sigma_F)$.
\qed

\medskip
\subsection*{Proof of Lemma \ref{lemma_rat_seq} }
Define iteratively six sequences as follows. Let $\thetalow^0 = \xilow^0 = \xnlow^0 = -\infty$ and $\thetaup^0 = \xiup^0 = \xnup^0 = \infty$, and for $k \geq 1$, 
\begin{equation*}
    \begin{cases}
        \thetalow^k = \br_L(\xilow^{k-1}) & \\
        \thetaup^k = \br_L(\xiup^{k-1}) & \\
        \xilow^k = \br_F^\invest(\thetaup^k, \xilow^{k-1}) & \\
        \xiup^k = \br_F^\invest(\thetalow^k, \xiup^{k-1}) & \\
        \xnlow^k = \br_F^\notinvest(\thetaup^k, \xnlow^{k-1}) & \\
        \xnup^k = \br_F^\notinvest(\thetalow^k, \xnup^{k-1})
    \end{cases},
\end{equation*}
where $\theta = \br_L(x)$, $x \in \RR \cup \{-\infty, \infty\}$, is 
the unique solution to
\begin{equation*}
   \pi_L(\theta; x) = \theta - \Phi\left(\frac{x - \theta}{\sigma_F}\right) = 0,
\end{equation*}
and, $\br_F^h(\theta', x')$, $(\theta', x') \in \RR \times \RR \cup \{-\infty, \infty\}$, is the unique value of $x$, if exists, that solves
\begin{equation*} 
     \pi_F^h(x; \theta', x') = \EE_{\theta \sim \Psi^h(\cdot; \,x, \theta')}\left[ \theta - \frac{n-1}{n} \Phi\left( \frac{x' - \theta}{\sigma_F} \right) \right] - \frac{\chi_\notinvest}{n} = 0; 
\end{equation*}
otherwise
\begin{equation*}
    \br_F^h(\theta', x') = \begin{cases}
        -\infty & \text{if $h = \invest$} \\
        \infty & \text{if $h = \notinvest$}
    \end{cases}.
\end{equation*}
Now we prove statements (a)-(f) as follows by induction.

\vskip 0.5 \baselineskip
\noindent Parts (a), (b), (c) \& (d): For $k = 1$, 
we have $\thetalow^1 = 0$ and $1 = \thetaup^1 < \thetaup^0$
because $\xilow^0 = -\infty$ and $\xiup^0 = \infty$. It follows that $\xilow^1 = \br_F^\invest(1, -\infty) = - \infty$ because $\EE_{\theta \sim \Psi^\invest(\cdot; \,x, 1)}[\theta] > 0$ for all $x$ by Lemma \ref{lemma_truncated_expectation}. Also, $\xiup^1 = \br_F^\invest(0, \infty)$ is the unique solution to $\EE_{\theta \sim \Psi^\invest(\cdot; \,\xiup^1, 0)}[\theta] = (n-1)/n$; therefore $\xiup^1 < \xiup^0$.

Suppose now that $\thetalow^k = 0$, $\thetaup^k < \thetaup^{k-1}$, $\xilow^k = -\infty$, and $\xiup^k < \xiup^{k -1}$ for any given $k \geq 2$. Then $\thetalow^{k+1} = \br_L(\xilow^k) = \br_L(-\infty) =  0$, and $0 < \thetaup^{k+1} = \br_L(\xiup^k) < \br_L(\xiup^{k-1}) = \thetaup^k$ because the leader's payoff is strictly decreasing in $x_\invest$. Furthermore, 
$\xilow^{k+1} = \br_F^\invest(\thetaup^{k+1}, \xilow^k) = \br_F^\invest(\thetaup^{k+1}, -\infty) = -\infty$
because $\EE_{\theta \sim \Psi^\invest(\cdot; \,x, \thetaup^{k+1})}[\theta] > \EE_{\theta \sim \Psi^\invest(\cdot; \,x, 0}[\theta] > 0$ for all $x$ by Lemma \ref{lemma_truncated_expectation}. Since Lemma \ref{lemma_x_payoff} implies that $\br_F^\invest(0, x)$ is strictly increasing in $x$, it is now clear that
\begin{equation*}
    \xiup^{k+1} = \br_F^\invest(\thetalow^{k+1}, \xiup^k) = \br_F^\invest(0, \xiup^k) < \br_F^\invest(0, \xiup^{k-1}) = \br_F^\invest(\thetalow^k, \xiup^{k-1}) = \xiup^k.
\end{equation*}
This completes the proof of parts (a), (b), (c), and (d).

\vskip 0.5 \baselineskip
\noindent Part (e): For $k = 1$, $\xnlow^1 = \br_F^\notinvest(\thetaup^1, \xnlow^0) = \br_F^\notinvest(\thetaup^1, -\infty) > \xnlow^0$ because, by Lemma \ref{lemma_truncated_expectation}, $\EE_{\theta \sim \Psi^\notinvest(\cdot; \,x, \thetaup^1)}[\theta] = 1/n$ has a unique solution.
Suppose now that $\xnlow^k > \xnlow^{k-1}$ for any given $k \geq 2$.
We know from Lemma \ref{lemma_x_payoff} that $\br_F^\notinvest(\theta, x)$ is strictly decreasing in $\theta$ but strictly increasing in $x$. Thus, $\xnlow^{k+1} = \br_F^\notinvest(\thetaup^{k+1}, \xnlow^k) > \br_F^\notinvest(\thetaup^k, \xnlow^k) > \br_F^\notinvest(\thetaup^k, \xnlow^{k-1}) = \xnlow^k$.

\vskip 0.5 \baselineskip
\noindent Part (f): For $k=1$, we know by Lemma \ref{lemma_truncated_expectation} that $\pi_F^\notinvest(x; \thetalow^1, \xnup^0) = \EE_{\theta \sim \Psi^\notinvest(\cdot; \,x, 0)}[\theta] - 1 < -1$ for all $x$. This implies that $\xnup^1 = \infty$.
Suppose that $\xnup^k = \infty$ for any given $k \geq 2$. We again have $\pi_F^\notinvest(x; \thetalow^k, \xnup^0) = \EE_{\theta \sim \Psi^\notinvest(\cdot; \,x, 0)}[\theta] - 1 < -1
$ for all $x$. Thus, $\xnup^{k+1} = \infty$. \qed

\medskip
\subsection*{Proof of Proposition \ref{prop_unique_rat} }
The model has a unique $\Delta$-rationalizable behavior if and only if $\thetaup = 0$, $\xiup = -\infty$, and $\xnlow = \infty$.
If $\xiup = -\infty$, Equation (\ref{leader_upper}) implies that
$\thetaup = 0$. By Equation (\ref{follower_notinvest_lower}),
it follows that $\xnup = \infty$ because we know from Equations (\ref{eq_expectation_of_theta}) and (\ref{follower_payoff_rank_belief}) that
\begin{equation*}
    \pi_F^\notinvest(x; 0, x) = x - \sigma_F \lambda\left( -\frac{x}{\sigma_F} \right) - \frac{n-1}{2n} \Phi \left(\frac{x}{\sigma_F}\right) - \frac{n+1}{2n} < 0
\end{equation*}
for all $x$. However, if there exists a value of $x$ such that $\pi_F^\invest(x; 0, x) = 0$, then 
\begin{equation*}
    \xiup = \max \left\{x \vl \pi_F^\invest(x; 0, x) =0 \right\}
\end{equation*}
and hence $\thetaup > \thetalow$. Thus, a unique $\Delta$-rationalizable behavior obtains if and only if $\xiup = -\infty$. It is worth noting that 
$\xiup > -\infty$ does not necessarily imply
that $\xnlow < \infty$.

Now we show that there exists a unique $\widehat{\sigma}_F$ such that $\xiup = -\infty$ if and only if $\sigma_F > \widehat{\sigma}_F$.
Again by Equations (\ref{eq_expectation_of_theta}) and (\ref{follower_payoff_rank_belief}) we can write $\pi_F^\invest(x; 0, x) = 0$ as
\begin{equation} \label{app_eq_fb}
    x = \frac{n-1}{2n}\Phi\left(\frac{x}{\sigma_F}\right) - \sigma_F \lambda\left(\frac{x}{\sigma_F}\right). \tag{A.1}
\end{equation}
Define $\rho(x, \sigma_F)$ to be the right-hand side of Equation (\ref{app_eq_fb}).
Observe that $\partial \rho(x, \sigma_F)/\partial x > 0$, $\lim_{x \to \infty}\rho(x, \sigma_F) = (n-1)/(2n)$, $\lim_{x \to -\infty} \rho(x, \sigma_F) = 1$, and $x > \rho(x, \sigma_F)$ for $x < 0$ but $|x|$ sufficiently large. 
Let $\bar{\sigma}_F = (n-1)/(8n\phi(0))$. It suffices to consider the following two cases.

\textit{Case 1}: Suppose that $\sigma_F \leq \bar{\sigma}_F$. It is equivalent to $\rho(0, \sigma_F) \geq 0$. Since
\begin{equation*}
    \frac{\partial \rho(0, \sigma_F)}{\partial x} = \frac{n-1}{2n\sigma_F}\phi(0) - \lambda'(0) \geq 8\phi(0)^2 > 1
\end{equation*}
by the derivative formula $\lambda'(x) = -\lambda(x)\left[ x + \lambda(x) \right]$,
Equation (\ref{app_eq_fb}) must have at least two solutions.

\textit{Case 2}: Suppose that $\sigma_F > \bar{\sigma}_F$.
For $x > 0$, 
\begin{equation} \label{app_rho_sig}
   \frac{\partial \rho(x, \sigma_F)}{\partial \sigma_F} = - \frac{x}{\sigma_F} \left(\frac{\partial \rho(x, \sigma_F)}{\partial x}\right) < 0, \tag{A.2}
\end{equation}
and
\begin{equation} \label{app_rho_sd}
    \frac{\partial^2 \rho(x, \sigma_F)}{\partial x^2} = \frac{n-1}{2n\sigma_F^2}\phi'\left(\frac{x}{\sigma_F}\right) - \frac{1}{\sigma_F}\lambda''\left(\frac{x}{\sigma_F}\right) < 0 \tag{A.3}
\end{equation}
because $\lambda''(x) > 0$. By (\ref{app_rho_sig}) and (\ref{app_rho_sd}), there exists a unique $\widehat{\sigma}_F$ such that $x$ and $\rho(x, \widehat{\sigma}_F)$ are tangent to each other at some $x > 0$ because $\rho(x, \sigma_F)$ is strictly concave. Moreover, for $x > 0$, $x > \rho(x, \sigma_F)$ if and only if $\sigma_F > \widehat{\sigma}_F$. If we can prove that $x > \rho(x, \sigma_F)$ for $x \leq 0$ whenever $\sigma_F > \widehat{\sigma}_F$, then we are done. Note that, at $\sigma_F = \bar{\sigma}_F$, 
\begin{equation*}
    \rho(x, \bar{\sigma}_F) - x = \bar{\sigma}_F \bigg[2\lambda(0)\Phi(z) - \lambda(z) - z \bigg] < 0,
\end{equation*}
where $z = x/\bar{\sigma}_F$ because $\lambda(0) = 2\phi(0)$ and $2\lambda(0)\Phi(x) - \lambda(x) - x < 0$ for $x < 0$. It follows that $x > \rho(x, \sigma_F)$ for all $x \leq 0$. Thus, (\ref{app_eq_fb}) can 
only have solutions if $\sigma_F \leq \widehat{\sigma}_F$.

Taking Case 1 and Case 2 together, we can conclude that (\ref{app_eq_fb}) has no solution (i.e., $\xiup = -\infty$) if and only if $\sigma > \widehat{\sigma}_F$. \qed

\medskip
\subsection*{Proof of Proposition \ref{prop_limit}}

\underline{Step 1}. We first show that $\xiup \to (n-1)/(2n)$ as $\sigma_F \to 0$. 
For sufficiently small $\sigma_F$,
note that $\xiup > -\infty$ is determined by Equation (\ref{follower_invest_upper}); that is,
\begin{equation*}
    \xiup + \sigma_F \lambda \left( \frac{\xiup - \thetalow}{\sigma_F}  \right) - \frac{n-1}{2n} \Phi\left( \frac{\xiup - \thetalow}{\sigma_F}  \right) = 0,
\end{equation*}
where $\thetalow = 0$.
We know from the proof of Proposition \ref{prop_unique_rat} that $\xiup > 0$. So, as $\sigma_F \to 0$, there are three cases to consider: (i) $\xiup \to 0$ and $\xiup/\sigma_F \to k \geq 0$, (ii) $\xiup \to 0$ and $\xiup/\sigma_F \to \infty$, and  (iii) $\xiup \to \aleph > 0$ and $\xiup/\sigma_F \to \infty$, 

\emph{Case (i)}: Suppose that $\xiup \to 0$ and $\xiup/\sigma_F \to k$, where $k \geq 0$ is a constant. It follows that
\begin{equation*}
    \xiup + \sigma_F \lambda \left( \frac{\xiup}{\sigma_F}  \right) - \frac{n-1}{2n} \Phi\left( \frac{\xiup}{\sigma_F}  \right) \to 0 + 0 \cdot \lambda(k) - \frac{n-1}{2n}\Phi(k) \neq 0,
\end{equation*}
which leads to a contradiction.

\emph{Case (ii)}: Suppose that $\xiup \to 0$ and $\xiup/\sigma_F \to \infty$. Then we have a contradiction because
\begin{equation*}
    \xiup + \sigma_F \lambda \left( \frac{\xiup}{\sigma_F}  \right) - \frac{n-1}{2n} \Phi\left( \frac{\xiup}{\sigma_F}  \right) \to - \frac{n-1}{2n} < 0.
\end{equation*}

\emph{Case (iii)}: Suppose that $\xiup \to \aleph > 0$. Then it must be the case that
\begin{equation*}
    \xiup + \sigma_F \lambda \left( \frac{\xiup}{\sigma_F}  \right) - \frac{n-1}{2n} \Phi\left( \frac{\xiup}{\sigma_F}  \right) \to \aleph - \frac{n-1}{2n} = 0,
\end{equation*}
which results in a contradiction except for the case $\aleph = (n-1)/(2n)$.

Combining all three cases above, we can conclude that $\xiup \to (n-1)/(2n)$ as $\sigma_F \to 0$.

\vskip 0.5 \baselineskip
\noindent \underline{Step 2}. We show next that $\thetaup \to (n-1)/(2n)$ as $\sigma_F \to 0$. We consider, for the sake of contradiction, the following two cases: (i) $\thetaup \to  \tau < (n-1)/(2n)$, and (iii) $\thetaup \to  \tau > (n-1)/(2n)$.

\emph{Case (i)}: Suppose that $\thetaup \to  \tau$, where $\tau \in [0, (n-1)/(2n))$ is a constant. Note that $\thetaup$ is given by Equation (\ref{leader_upper}):
\begin{equation*}
    \thetaup - \Phi \left( \frac{\xiup - \thetaup}{\sigma_F} \right) = 0.
\end{equation*}
Since $(\xiup - \thetaup)/\sigma_F \to \infty$, we have
\begin{equation*}
\thetaup - \Phi \left( \frac{\xiup - \thetaup}{\sigma_F} \right) \to \tau - 1 < - \frac{n+1}{2n} < 0,     
\end{equation*}
which leads to a contradiction.

\emph{Case (ii)}: Suppose that $\thetaup \to  \tau \in ((n-1)/(2n), 1]$. Then it must be the case that $(\xiup - \thetaup)/\sigma_F \to -\infty$. But we have a contradiction because 
\begin{equation*}
\thetaup - \Phi \left( \frac{\xiup - \thetaup}{\sigma_F} \right) \to \tau > \frac{n-1}{2n} > 0.     
\end{equation*}

Thus, we must have $\thetaup \to (n-1)/(2n)$ as $\sigma_F \to 0$.

\medskip
\noindent \underline{Step 3}. Lastly, we show that $\xnlow \to \infty$ as $\sigma_F \to 0$. 
By way of contradiction, suppose that $\xnlow \to \aleph < \infty$. This means that $\aleph$ solves  
Equation (\ref{follower_notinvest_lower}):
\begin{equation*}
    \aleph - \sigma_F \lambda \left( \frac{\thetaup - \aleph}{\sigma_F} \right) + \frac{n-1}{2n} \Phi\left( \frac{\thetaup - \aleph}{\sigma_F}\right) = 1.
\end{equation*}
There are three possible cases for the limit of $(\thetaup - \xnlow)/\sigma_F$ as $\sigma_F \to 0$: (i) $-\infty$, (ii) $\infty$, or (iii) a constant $k \in \RR$.

\emph{Case (i)}: Suppose that $(\thetaup - \xnlow)/\sigma_F \to -\infty$. Since $\lim_{x \to -\infty} x + \lambda(x) = 0$, we have
\begin{align*}
    \xnlow - & \sigma_F \lambda \left( \frac{\thetaup - \xnlow}{\sigma_F} \right) + \frac{n-1}{2n} \Phi\left( \frac{\thetaup - \xnlow}{\sigma_F}\right) \\
    & = -\sigma_F \left( \frac{\thetaup - \xnlow}{\sigma_F} + \lambda \left( \frac{\thetaup - \xnlow}{\sigma_F} \right) \right) + \thetaup + \frac{n-1}{2n} \Phi\left( \frac{\thetaup - \xnlow}{\sigma_F}\right) \\
    & \to \frac{n-1}{2n} < 1.
\end{align*}
Thus, we have a contradiction.

\emph{Case (ii)}: Suppose that $(\thetaup - \xnlow)/\sigma_F \to \infty$. In this case, we must have $\aleph \leq (n-1)/(2n)$. It follows that
\begin{equation*}
    \xnlow - \sigma_F \lambda \left( \frac{\thetaup - \xnlow}{\sigma_F} \right) + \frac{n-1}{2n} \Phi\left( \frac{\thetaup - \xnlow}{\sigma_F}\right) \to \aleph + \frac{n-1}{2n} \leq \frac{n-1}{n} < 1,
\end{equation*}
which leads to a contradiction.

\emph{Case (iii)}: Suppose that $(\thetaup - \xnlow)/\sigma_F \to k$, where $k$ is a constant. This implies that $\aleph = (n-1)/(2n)$. But since
\begin{equation*}
    \xnlow - \sigma_F \lambda \left( \frac{\thetaup - \xnlow}{\sigma_F} \right) + \frac{n-1}{2n} \Phi\left( \frac{\thetaup - \xnlow}{\sigma_F}\right) \to \aleph + \frac{n-1}{2n}\Phi(k) < \frac{n-1}{n} < 1,
\end{equation*}
we have a contradiction.

Thus, we must have $\xnlow \to \infty$ as $\sigma_F \to 0$. This completes the proof. \qed

\medskip
\subsection*{Derivation of the conditional rank beliefs in (\ref{rank_belief}) } 
By definition, we have
\begin{align*}
    R^h(x; z) & = \prob(x_k \leq x_j \vl x_j = x, z) \\
    & = \prob(\theta + \sigma_F \varepsilon_k \leq x_j \vl x_j = x, z) \\
    & = \int_{-\infty}^\infty \left( \int_{-\infty}^\frac{x - \theta}{\sigma_F} \phi(\varepsilon) \dd \varepsilon \right) \dd \Psi^h(\theta; x, z) \\
    & = \int_{-\infty}^\infty \Phi\left( \frac{x - \theta}{\sigma_F} \right) \dd \Psi^h(\theta; x, z).
\end{align*}
Now under $h = \invest$, by (\ref{interim_belief})
\begin{align*}
    R^\invest(x; z) & = \frac{1}{\Phi\left( \frac{x - z}{\sigma_F}\right) } \int_z^\infty \Phi\left( \frac{x - \theta}{\sigma_F} \right) \frac{1}{\sigma_F} \phi\left( \frac{\theta - x}{\sigma_F}  \right) \dd \theta \\
    & = \frac{1}{\Phi\left( \frac{x - z}{\sigma_F}\right) } \int_{-\infty}^{\Phi\left(\frac{x-z}{\sigma_F}\right)} \eta \dd \eta \\ 
    & = \frac{1}{2}\Phi\left( \frac{x-z}{\sigma_F} \right).
\end{align*}
The second equality is given by a change of variable $\eta = \Phi((x - \theta)/\sigma_F)$. Similarly,
\begin{align*}
    R^\notinvest(x; z) & = \frac{1}{\Phi\left( \frac{z - x}{\sigma_F}\right) } \int_{-\infty}^z \Phi\left( \frac{x - \theta}{\sigma_F} \right) \frac{1}{\sigma_F} \phi\left( \frac{\theta - x}{\sigma_F} \right) \dd \theta \\
    & = \frac{1}{\Phi\left( \frac{z - x}{\sigma_F}\right) } \int_{\Phi\left(\frac{x-z}{\sigma_F}\right)}^1 \eta \dd \eta \\
    & = \frac{1}{2} \Phi\left( \frac{x-z}{\sigma_F} \right) + \frac{1}{2}.
\end{align*}

\medskip
\subsection*{Derivation of the conditional rank beliefs in (\ref{ext_rank_belief})}

Under history $h = \invest$,
\begin{align*}
    R^\invest(x; z) & = \prob(x_k \leq x_j \vl x_j = x, x_L > z) \\
    & = \int_{-\infty}^\infty \left( \int_{-\infty}^\frac{x - \theta}{\sigma_F} \phi(\varepsilon) \dd \varepsilon \right) g^\invest (\theta; x , z) \dd \theta \\ 
    & = \frac{1}{\Phi \left( \frac{x - z }{\sigma} \right) }\int_{-\infty}^\infty \Phi(-y) \phi(y) \Phi \left( \frac{x - z + \sigma_F y }{\sigma_L} \right) \dd y \\
    & = \frac{1}{2} - \frac{1}{\Phi \left( \frac{x - z }{\sigma} \right)} T \left( \frac{x - z}{\sigma},  \frac{\sigma_F}{(2\sigma_L^2 + \sigma_F^2)^{\frac{1}{2}}} \right).
\end{align*}
The second equality is given by the independence between 
$\theta$ and $\varepsilon_k$,
the third equality is due to a change of variable $y = (\theta - x_\invest)/\sigma_F$, and the last equality is derived from applying the following integral identity
\begin{equation*}
    \int_{-\infty}^\infty \Phi(a + bx) \Phi(cx) \phi(x) \dd x = \frac{1}{2} \Phi \left( \frac{a}{\sqrt{1 + b^2}} \right) + T \bigg( \frac{a}{\sqrt{1 + b^2}}, \frac{bc}{\sqrt{1 + b^2 + c^2}}\bigg)
\end{equation*}
and $T(y, -a) = -T(y, a)$, where $T(y, a)$ is Owen's T-function (see footnote \ref{owen_t}).

Under history $h = \notinvest$, a similar argument yields
\begin{align*}
    R^\notinvest(x; z) & = \prob (x_k \leq x_j \vl x_j = x, x_L \leq z) \\
    & = \int_{-\infty}^\infty \Phi \left( \frac{x - \theta }{\sigma_F} \right)  g^\notinvest(\theta; x, z) \dd \theta \\
    & = \frac{1}{2} + \frac{1}{\Phi \left( \frac{z - x}{\sigma} \right)} T \left( \frac{z - x}{\sigma},  \frac{\sigma_F}{(2\sigma_L^2 + \sigma_F^2)^{\frac{1}{2}}} \right).
\end{align*}

\medskip

\begin{lemma} \label{lemma_exp_ext}
Under history $h$, $\EE_{\theta \sim G^h(\cdot; \,x, z)}[\theta]$ is increasing in $x$ and $z$. Moreover,
\[
\EE_{\theta \sim G^h(\cdot; \,x, z)}[\theta] \to \begin{cases}
     \infty & \text{as $x \to \infty$} \\
     -\infty & \text{as $x \to -\infty$}
\end{cases}.
\]
\end{lemma}
\begin{proof}[Proof of Lemma \ref{lemma_exp_ext}]
Note that under history $h = \invest$,
\begin{align} \label{app_fol_exp}
    \EE_{\theta \sim G^\invest(\cdot; \,x, z)} [\theta] & = \frac{1}{\Phi\left(\frac{x-z}{\sigma}\right)}\int_{-\infty}^\infty \frac{t}{\sigma_F} \phi\left(\frac{t-x}{\sigma_F}\right) \Phi\left(\frac{t-z}{\sigma_L}\right) \dd t  \nonumber \\
    & = \frac{1}{\Phi\left(\frac{x-z}{\sigma}\right)}\int_{-\infty}^\infty (x + \sigma_F \eta) \phi(\eta) \Phi\left(\frac{x + \sigma_F \eta - z}{\sigma_L}\right) \dd \eta \nonumber \\
    & = x + \frac{\sigma_F^2}{\sigma} \lambda \left(\frac{x-z}{\sigma}\right). \tag{A.4}
\end{align}
The second inequality is due to a change of variable $\eta = (t-x)/\sigma_F$, and the third equality is derived by applying integral identities $\int_{-\infty}^\infty \phi(\eta)\Phi(a + b\eta) \dd \eta = \Phi(a/\sqrt{1 + b^2})$ and $\int_{-\infty}^\infty \eta \phi(\eta) \Phi(a+b\eta) \dd \eta = (b/\sqrt{1+b^2})\Phi(a/\sqrt{1+b^2})$ (see, for example, \cite{owen_1980}). Therefore, (\ref{app_fol_exp}) is increasing in $x$ and $z$ because
\begin{align*}
    \EE_{\theta \sim G^\invest(\cdot; \,x, z)} [\theta] & = \frac{\sigma_L^2}{\sigma}\left(\frac{x-z}{\sigma}\right) + \frac{\sigma_F^2}{\sigma}\left(\frac{x-z}{\sigma} + \lambda\left(\frac{x-z}{\sigma}\right)\right) + z,
\end{align*}
$\eta + \lambda(\eta)$ is increasing in $\eta$, and $-1 < \lambda'(\eta) < 0$. The last inequality is due to \cite{sampford_1953}.
Now given the fact that $\lambda(\eta)/\eta \to -1$ as $\eta \to -\infty$ and $\lambda(\eta)/\eta \to 0$ as $\eta \to \infty$, we can conclude that
\[
\EE_{\theta \sim G^\invest(\cdot; \,x, z)} [\theta] \to \begin{cases}
    \infty & \text{as $x \to \infty$} \\
    -\infty & \text{as $x \to -\infty$}
\end{cases}.
\]

Similarly, under history $h = \notinvest$, one can show that
\begin{align*}
    \EE_{\theta \sim G^\notinvest(\cdot; \,x, z)}[\theta] & = x - \frac{\sigma_F^2}{\sigma} \lambda\left(\frac{z-x}{\sigma}\right) \\
    & = -\frac{\sigma_L^2}{\sigma}\left(\frac{z-x}{\sigma}\right) - \frac{\sigma_F^2}{\sigma}\left(\frac{z-x}{\sigma} + \lambda\left(\frac{z-x}{\sigma}\right)\right) + z,
\end{align*}
which is increasing in $x$ and $z$, approaches to $-\infty$ as $x \to -\infty$, and approaches to $\infty$ as $x \to \infty$. The proof is complete.
\end{proof}

\begin{lemma} \label{lemma_payoff_ext}
Under history $h$, $\pi_F^h(x; z, x_h)$ is increasing in $x$ and $z$, and is decreasing in $x_h$. Moreover, $\lim_{x \to -\infty}\pi_F^h(x; z, x_h) = -\infty$ and $\lim_{x \to \infty} \pi_F^h(x; z, x_h) = \infty$.
\end{lemma}
\begin{proof}[Proof of Lemma \ref{lemma_payoff_ext}]
Recall that
\[
\pi_F^h(x; z, x_h) = \EE_{\theta \sim G^h(\cdot; \,x, z)} \left[\theta - \frac{n-1}{n}\Phi\left(\frac{x_h - \theta}{\sigma_F}\right)\right] - \frac{\chi_\notinvest}{n}.
\]
It is immediate to see that $\pi_F^h(x; z, x_h)$ is decreasing in $x_h$ since $\Phi((x_h - \theta)/\sigma_F)$ is increasing in $x_h$.
To show $\pi_F^h(x; z, x_h)$ is increasing in both $x$ and $z$, it suffices to prove that $G^h(\cdot; \,x, z)$ is increasing in $x$ and $z$ with respect to the first-order stochastic dominance order. We ignore the proof as it is similar to that of Lemma \ref{lemma_x_payoff}.
The limits at infinity are given directly by Lemma \ref{lemma_exp_ext} and the boundedness of $\Phi((x_h - \theta)/n)$.
\end{proof}

\begin{lemma} \label{lemma_seq_ext} 
Let $\xlow^0 = \xhlow^0 = -\infty$ and $\xup^0 = \xhup^0 = \infty$ for each history $h$. Then
the iterative procedure of $\Delta$-rationalizability yields six sequences: \\
(i) $(\xlow^k)_{k = 0}^\infty, (\xilow^k)_{k = 0}^\infty$, and $(\xnlow^k)_{k = 0}^\infty$ are strictly increasing and bounded above;\\
(ii) $(\xup^k)_{k = 0}^\infty, (\xiup^k)_{k = 0}^\infty$, and $(\xnup^k)_{k = 0}^\infty$ are strictly decreasing and bounded below.
\end{lemma}
\begin{proof}[Proof of Lemma \ref{lemma_seq_ext}]
Let $x_L = \br_L(x_\invest)$ denote the unique solution to
\[
\pi_L(x_L; x_\invest) = x_L - \Phi\left(\frac{x_\invest-x_L}{\sigma}\right),
\]
and $x = \br_F^h(z, x_h)$ the unique solution to $\pi_F^h(x; z, x_h) = 0$ for each history $h$.
The latter is guaranteed by both Lemma \ref{lemma_exp_ext} and Lemma \ref{lemma_payoff_ext}.

Let $\xlow^0 = \xhlow^0 = -\infty$ and $\xup^0 = \xhup^0 = \infty$. Define for $k \in \NN$,
\begin{equation} \label{ext_seq_br}
\begin{cases}
\xlow^k = \br_L(\xilow^{k-1}) \\
\xup^k = \br_L(\xiup^{k-1}) \\
\xilow^k = \br_F^\invest(\xup^k, \xilow^{k-1}) \\
\xiup^k  = \br_F^\invest(\xlow^k, \xiup^{k-1}) \\
\xnlow^k = \br_F^\notinvest(\xup^k, \xnlow^{k-1}) \\
\xnup^k = \br_F^\notinvest(\xlow^k, \xnup^{k-1}) \\
\end{cases}. \tag{A.5}
\end{equation}
Note that, by Lemma \ref{lemma_payoff_ext}, $\br_L(x_\invest)$ is strictly increasing in $x_\invest$ and $\br_F^h(z, x_h)$ is strictly decreasing in $z$ and strictly increasing in $x_h$.  We prove the lemma by induction. 

\noindent For $k = 1$, consider the leader first. Since $\xlow^1 = \br_L(\xilow^0) = 0$ and $\xup^0 = \br_L(\xiup^0) = 1$, it is immediate that $\xlow^0 < \xlow^1 < \xup^1 < \xup^0$. Now under history $h$, Lemma \ref{lemma_exp_ext} implies that $\xhlow^1 > \xhlow^0$ and $\xhup^1 < \xhup^0$. But since
\[
\xhlow^1 = \br_F^h(\xup^1, \xhlow^0) < \br_F^h(\xlow^1, \xhlow^0) < \br_F^h(\xlow^1, \xhup^0) = \xhup^1,
\]
it follows that $\xhlow^0 < \xhlow^1 < \xhup^1 < \xhup^0$.

\noindent Assume that, for $k \geq 1$, $\xlow^{k-1} < \xlow^k < \xup^k < \xup^{k-1}$ for the leader, and $\xhlow^{k-1} < \xhlow^k < \xhup^k < \xhup^{k+1}$ for followers under history $h$. The induction hypothesis implies
\begin{align*}
    \br_L(\xilow^{k-1}) < \br_L(\xilow^k) < \br_L(\xiup^k) < \br_L(\xiup^{k-1});
\end{align*}
therefore we have $\xlow^k < \xlow^{k+1} < \xup^{k+1} < \xup^k$.
The proof for $\xhlow^k < \xhlow^{k+1} < \xhup^{k+1} < \xhup^k$ is straightforward.
\end{proof}

\medskip
\noindent Define $\tilde{S}(y, \nu) = 1/2 - T(y, \nu)/\Phi(y)$, where $T(y, \nu)$ is Owen's T-function with slope parameter $\nu > 0$ and $y \in \RR$ (see footnote \ref{owen_t}).
\begin{lemma} \label{lemma_s_func}
    The function $\tilde{S}(y, \nu)$ is strictly increasing and differentiable everywhere in $y$. Moreover, $\lim_{y \to -\infty} \tilde{S}(y, \nu) = 0$, $\lim_{y \to \infty} \tilde{S}(y, \nu) = 1/2$, and if $\nu \in (0, 1)$, then $\partial \tilde{S}(y, \nu)/\partial y$ vanishes at infinity ; i.e.,
    \begin{equation*}
         \lim_{y \to -\infty} \frac{\partial  \tilde{S}(y, \nu)}{\partial y} =          \lim_{y \to \infty} \frac{\partial  \tilde{S}(y, \nu)}{\partial y} = 0.
    \end{equation*}
\end{lemma}
\begin{proof}[Proof of Lemma \ref{lemma_s_func}]
Since $T(y, \nu)$ is differentiable everywhere in $y$, so is $\tilde{S}(y, \nu)$. Note that
\begin{align*}
    \frac{\dd }{\dd y}\left[ \frac{T(y, \nu)}{\Phi(y)} \right] & = \frac{1}{2\Phi(y)^2} \left[-\phi(y) \erf\left( \frac{\nu y}{\sqrt{2}} \right) \Phi(y) - 2\phi(y) T(y, \nu)\right]  \\
 & \propto - \phi(y) \erf\left( \frac{\nu y}{\sqrt{2}} \right) \Phi(y) - \phi(y) \int_{-y}^\infty \phi(t) \erf\left(
 \frac{\nu t}{\sqrt{2}} \right) \dd t \\
 & = - \phi(y) \int_{-y}^\infty \phi(t) \left[ \erf\left( 
 \frac{\nu y}{\sqrt{2}} \right) + \erf\left( 
 \frac{\nu t}{\sqrt{2}} \right) \right] \dd t \\
 & < 0,
\end{align*}
where $\erf(\cdot)$ is the error function.
The second equality is derived using the integral representation of Owen's T-function (see Equation (3.1) in \cite{brychkov_savischenko_2016}). The inequality holds because the strict monotonicity of $\erf(\cdot)$ implies that
\begin{equation*}
    \erf\left( \frac{\nu y}{\sqrt{2}} \right) + \erf\left( 
 \frac{\nu t}{\sqrt{2}} \right) >  \erf\left( 
 \frac{\nu y}{\sqrt{2}} \right) + \erf\left( 
 -\frac{\nu y}{\sqrt{2}} \right) = 0
\end{equation*}
for all $t > -y$. Thus, $\tilde{S}(y, \nu)$ is strictly increasing. Moreover, since $T(y, \nu) \to 0$ as $y \to \infty$ or $y \to -\infty$, we have $\lim_{y \to \infty} \tilde{S}(y, \nu) = 1/2$ and
\begin{equation*}
    \lim_{y \to -\infty} \tilde{S}(y, \nu) = \frac{1}{2} -\lim_{y \to -\infty}\frac{T(y, \nu)}{\Phi(y)} = \frac{1}{2} + \frac{1}{2}\lim_{y \to -\infty} \erf\left( \frac{\nu y}{\sqrt{2}} \right)  = 0.
\end{equation*}
The last equality is due to $\erf(\nu y/\sqrt{2}) \to -1$ as $y \to -\infty$.

Now we define
\begin{equation*}
    M(y, \nu) = \int_{-y}^\infty \phi(t) \erf\left( 
 \frac{\nu t}{\sqrt{2}} \right)  \dd t + \erf\left( \frac{\nu y}{\sqrt{2}}\right)\Phi(y),
\end{equation*}
and write $\tilde{S}_y(y, \nu)$ for the partial derivative $\partial \tilde{S}(y, \nu)/\partial y$. Then we know from above that
\begin{equation*}
    \tilde{S}_y(y, \nu) = \frac{\phi(y) M(y, \nu)}{2 \Phi(y)^2}.
\end{equation*}
To show that $\tilde{S}_y(y, \nu)$ vanishes as $y \to \infty$, note that
\begin{equation*}
    \lim_{y \to \infty} M(y, \nu) = \int_{-\infty}^\infty \phi(t) \erf\left( \frac{\nu t}{\sqrt{2}}\right) \dd t + 1
    = 2 \int_{-\infty}^\infty \phi(t) \Phi(\nu t) \dd t
    = 2 \Phi(0) = 1.
\end{equation*}
The first equality holds because $\erf(\nu y/\sqrt{2}) \to 1$ as $y \to \infty$, the second equality is due to $\erf(\nu t/\sqrt{2}) = 2 \Phi(\nu t) - 1$, and the last equality is given by applying the integral identity $\int_{-\infty}^\infty \phi(t) \Phi(a + bt) \dd t = \Phi(a/\sqrt{1+b^2})$. Thus, $\lim_{y \to \infty} \tilde{S}_y(y, \nu) = 0$. 

Suppose that $\nu \in (0, 1)$.
Since
\begin{equation*}
    M_y(y,\nu) = \sqrt{\frac{2\nu^2}{\pi}} \Phi(y) \exp{\left(-\frac{\nu^2 y^2}{2}\right)},
\end{equation*}
by L'H\^{o}pital's rule and $\lim_{y \to -\infty} \phi(y)/(y \Phi(y)) = -1$, we have
\begin{equation*}
    \lim_{y \to -\infty} \frac{y M(y, \nu)}{\Phi(y)} = \lim_{y \to -\infty} \frac{M(y, \nu)}{\Phi(y) y^{-1}} = \sqrt{\frac{2\nu^2}{\pi}} \lim_{y \to -\infty}  \frac{ \exp{\left(-\frac{\nu^2 y^2}{2}\right)}}{\frac{\phi(y)}{y \Phi(y)} - \frac{1}{y^2}} = 0.
\end{equation*}
This implies that
\begin{align*}
    \lim_{y \to -\infty} \frac{\phi(y) M(y, \nu)}{\Phi(y)^2} & = \lim_{y \to -\infty} \frac{-y M(y, \nu) + M_y(y, \nu)}{2 \Phi(y)} \\
    & = \sqrt{\frac{2\nu^2}{\pi}} \lim_{y \to -\infty} \exp{\left(-\frac{\nu^2 y^2}{2}\right)} \\
    & = 0.
\end{align*}
Thus, $\lim_{y \to -\infty} \tilde{S}_y(y, \nu) = 0$.
The proof is complete.
\end{proof}

\medskip
\subsection*{Proof of Proposition \ref{prop_unique_rat_ext}}
\emph{Proof of Part (i)}:
Fix $\sigma_F = \gamma \sigma_L$, $\gamma > 0$. Let $\xlow, \xup, \xilow, \xiup, \xnlow$, and $\xnup$ be the limits of the six sequences described in (\ref{ext_seq_br}), respectively. By Lemma \ref{lemma_seq_ext}, we know that they must solve (\ref{ext_sys_eqs}). Moreover $0 = \xlow^1 < \xlow \leq \xup < \xup^1 = 1$ and $\xhlow^1 < \xhlow \leq \xhup < \xhup^1$, where $\xhlow^1$ and $\xhup^1$ are the lower and upper dominance bounds in Round 1 for the followers under history $h$. Let $\Xi_L = [0, 1]$ and $\Xi_F^h = [\xhlow^1, \xhup^1]$.

Let $S(y) = \tilde{S}(y, \alpha)$ with slope parameter $\alpha = \gamma/\sqrt{2 + \gamma^2}$.
Since $1 + \gamma^2 (1 + \lambda'(y))$ is positive and bounded for all $y \in \RR$,
the following function is well-defined by Lemma \ref{lemma_s_func}:
\begin{equation*}
    \Lambda(\gamma) = \max_{y \in \RR} \frac{S'(y)}{1 + \gamma^2(1 + \lambda'(y))}.
\end{equation*}
Moreover, $\Lambda(\gamma) > 0$.  Now let
\begin{equation} \label{ext_sufficiency_1} 
     \widehat{\sigma}_L^1(\gamma) = \left(\frac{n-1}{n}\right)  \sqrt{(1+\gamma^2) \Lambda(\gamma)^2}. \tag{A.6}
\end{equation}
We prove this part in two steps. First, we show that if $\sigma_L > \widehat{\sigma}_L^1(\gamma)$
then the following system of equations with $x_L \in \Xi_L$ and $x_h \in \Xi_F^h$ 
\begin{equation} \label{ext_sys_small}
    \begin{cases}
        \pi_L(x_L;x_{\invest}) = 0 \\
        \pi_F^\invest(x_{\invest}; x_L, x_{\invest}) = 0 \\ 
        \pi_F^\notinvest(x_{\notinvest}; x_L,x_{\notinvest}) = 0 \\
    \end{cases} \tag{A.7}
\end{equation}
has a unique solution $(x_L^*, x_\invest^*, x_\notinvest^*)$. Second, we show that there exists $\widehat{\sigma}_L^2(\gamma)$ such that $(x_L^*, x_\invest^*, x_\notinvest^*)$ is also the unique solution to (\ref{ext_sys_eqs}) whenever
\begin{equation} \label{ext_suff_cond}
    \sigma_L > \widehat{\sigma}_L(\gamma) = \max \left\{\widehat{\sigma}_L^1(\gamma), \widehat{\sigma}_L^2(\gamma) \right\}. \tag{A.8}
\end{equation}

\noindent \underline{Step 1}: Assume from now on that $\sigma_L > \widehat{\sigma}_L^1(\gamma)$.
By Equations (\ref{ext_rank_belief}), (\ref{ext_follower_fp_eqn}), and (\ref{app_fol_exp}), we have
\begin{equation} \label{app_fol_fp_eqn}
    \pi_F^\invest(x_\invest; x_L, x_\invest) = x_\invest + \frac{\sigma_F^2}{\sigma} \lambda\left( \frac{x_\invest - x_L}{\sigma}\right) - \frac{n-1}{n}S\left( \frac{x_\invest - x_L}{\sigma} \right). \tag{A.9}
\end{equation}
If $\sigma_L > \widehat{\sigma}_L^1(\gamma)$, then $\sigma_F^2/\sigma^2 = \gamma^2/(1 + \gamma^2)$ gives that
\begin{align*}
    \frac{\partial \pi_F^\invest}{\partial x_\invest} & = 1 + \frac{\gamma^2}{1 + \gamma^2} \lambda'\left( \frac{x_\invest - x_L}{\sigma} \right) - \frac{n-1}{n\sigma} S'\left(\frac{x_\invest - x_L}{\sigma} \right)\\
    & > 1 + \frac{\gamma^2}{1 + \gamma^2} \lambda'\left( \frac{x_\invest - x_L}{\sigma} \right) - \frac{1}{(1 + \gamma^2) \Lambda(\gamma)} S'\left(\frac{x_\invest - x_L}{\sigma} \right) \geq 0.
\end{align*}
This implies that, for any given $x_L$, $\pi_F^\invest(x_\invest; x_L, x_\invest) = 0$ admits a unique solution. 
Similarly, since
\begin{equation*}
    \pi_F^\notinvest (x_\notinvest; x_L, x_\notinvest) = x_\notinvest - \frac{\sigma_F^2}{\sigma} \lambda\left( \frac{ x_L - x_\notinvest }{\sigma}\right) + \frac{n-1}{n}S\left( \frac{x_L - x_\notinvest}{\sigma} \right) - 1,
\end{equation*}
the condition $\sigma_L > \widehat{\sigma}_L^1(\gamma)$ implies that
\begin{align*}
    \frac{\partial \pi_F^\notinvest}{\partial x_\notinvest} & = 1 + \frac{\gamma^2}{1 + \gamma^2} \lambda'\left( \frac{ x_L - x_\notinvest }{\sigma}\right) - \frac{n-1}{n\sigma} S'\left(\frac{x_L - x_\notinvest}{\sigma} \right) \\
    & > 1 + \frac{\gamma^2}{1 + \gamma^2} \lambda'\left( \frac{x_L - x_\notinvest}{\sigma} \right) - \frac{1}{(1 + \gamma^2) \Lambda(\gamma)} S'\left(\frac{x_L - x_\notinvest}{\sigma} \right) \geq 0
\end{align*}
and hence, given $x_L$, $\pi_F^\notinvest (x_\notinvest; x_L, x_\notinvest) = 0$ has a unique solution.
Thus, (\ref{ext_sys_small}) has a unique solution if and only if there is a unique solution to its first two equations.

Since $\partial \pi_F^\invest/\partial x_\invest$ is continuous on $\Xi_F^\invest \times \Xi_L$, the extreme value theorem ensures that there exists $d_\invest > 0$ such that $\partial \pi_F^\invest/\partial x_\invest \geq d_\invest$. Thus, a global implicit function theorem (see, e.g., Lemma 2 of \cite{zhang_ge_2006}) implies that there is a unique function $f: \Xi_L \to \Xi_F^\invest$ such that $\pi_F^\invest(f(x_L); x_L, f(x_L)) = 0$. Moreover, $f \in C^1$ and $f' < 0$. It follows that
\begin{equation*}
    \frac{\dd \pi_L}{\dd x_L} = \frac{\partial \pi_L}{\partial x_L} + \frac{\partial \pi_L}{\partial x_\invest}f'(x_L) > 0.
\end{equation*}
This says that $\pi_L(x_L; f(x_L)) = 0$ has a unique solution. Thus, (\ref{ext_sys_small}) has a unique solution $(x_L^*, x_\invest^*, x_\notinvest^*)$.

\vskip 0.5 \baselineskip
\noindent \underline{Step 2}: We next show that there exists $\widehat{\sigma}_L^2(\gamma)$ such that
$(x_L^*, x_\invest^*, x_\notinvest^*)$ is the unique solution to (\ref{ext_sys_eqs}) if (\ref{ext_suff_cond}) holds.
Note that
\begin{equation*}
    \frac{\partial \pi_L}{\partial x_L} = 1 + \frac{1}{\sigma}\phi\left( \frac{x_\invest - x_L}{\sigma} \right) > 0
\end{equation*}
is continuous on $\Xi_F^\invest \times \Xi_L$, so there exists $d_L > 0$ such that $\partial \pi_L/\partial x_L \geq d_L$. Thus, there exists a global implicit function $g: \Xi_F^\invest \to \Xi_L$ such that $\pi_L(g(x_\invest); x_\invest) = 0$. In addition, we have $g \in C^1$ and $g' > 0$. 
Then the first four equations of (\ref{ext_sys_eqs}) implies that
\begin{equation*}
    \xlow = g(\xilow) = (g \circ f)(\xup) = (g\circ f \circ g) (\xiup) =  (g\circ f \circ g \circ f)(\xlow).
\end{equation*}
Consider $h: \Xi_L \to \Xi_L$ such that $h = g \circ f \circ g \circ f$. Clearly, $x_L^*$ is a fixed point of $h$
(shown in Step 1). We also note that $h(0) > 0$ and $h' > 0$ because $f' < 0$ and $g' > 0$.
Define
\begin{equation*}
    M(\sigma_L, \gamma) = \max_{y \in \RR} - \frac{\gamma^2}{1+\gamma^2} \lambda'(y) + \frac{n-1}{n\sigma_L \sqrt{1 + \gamma^2}} S'(y).
\end{equation*}
We have $M(\sigma_L, \gamma) < 1$ because $\sigma_L > \widehat{\sigma}_L^1(\gamma)$. It follows that
\begin{equation*}
    f' \geq \frac{M(\sigma_L, \gamma)}{M(\sigma_L, \gamma) - 1} = M_f(\sigma_L, \gamma).
\end{equation*}
By the envelope theorem, $\partial M(\sigma_L, \gamma)/\partial \sigma_L < 0$. Thus, $\partial M_f(\sigma_L, \gamma)/\partial \sigma_L > 0$.
By the definition of $g$, we have
\begin{equation*}
    g' \leq \frac{\phi(0)}{\sigma_L \sqrt{1 + \gamma^2} + \phi(0)} = M_g(\sigma_L, \gamma).
\end{equation*}
Moreover, $\partial M_g(\sigma_L, \gamma)/\partial \sigma_L < 0$.
It follows that
\begin{equation} \label{h_prime}
    h' \leq \left[M_g(\sigma_L, \gamma) \cdot M_f(\sigma_L,\gamma)\right]^2. \tag{A.10}
\end{equation}
The right-hand side of (\ref{h_prime}) is strictly decreasing in $\sigma_L$ and approaches zero as $\sigma_L \to \infty$. Thus, there exists $\widehat{\sigma}_L^2(\gamma)$ such that $h' < 1$ if $\sigma_L > \max\{\widehat{\sigma}_L^1(\gamma),  \widehat{\sigma}_L^2(\gamma)\}$ (i.e., (\ref{ext_suff_cond}) holds). It is immediate that $h$ has a unique fixed point when $h' < 1$. The proof of part (i) is complete.   

\medskip
\noindent \emph{Proof of Part (ii)}:
Fix $\sigma_F = \gamma \sigma_L$, $\gamma > 0$. We prove this part in three steps. We first show the existence of a monotone equilibrium characterized by a tuple of thresholds $(x_L^*, x_\invest^*, x_\notinvest^*)$ such that
$a_L = \invest$ if and only if $x_L > x_L^*$, and $s_j(\invest) = \invest$ if and only if $x_j > x_\invest^*$ and $s_j(\notinvest) = \invest$ if and only if $x_j > x_\notinvest^*$ for all followers $j \in F$. 
In Step 2, we show that $x_\invest^* < x_L^*$. We finally show, in Step 3, that there are multiple $\Delta$-rationalizable profiles in the limit as $\sigma_L \to 0$.

\vskip 0.5 \baselineskip
\noindent \underline{Step 1}.
Let $x_L = g(x_\invest)$ be the solution to $\pi_L(x_L; x_\invest) = 0$. Substituting $g(x_\invest)$ into Equation (\ref{app_fol_fp_eqn}) yields
\begin{equation*}
    \pi_F^\invest(x_\invest; g(x_\invest), x_\invest) = x_\invest + \frac{\sigma_F^2}{\sigma} \lambda\left( \frac{x_\invest - g(x_\invest)}{\sigma}\right) - \frac{n-1}{n}S\left( \frac{x_\invest - g(x_\invest)}{\sigma} \right).
\end{equation*}
Since $g(x_\invest) \in (0, 1)$, $\pi_F^\invest(x_\invest; g(x_\invest), x_\invest)\to -\infty$ as $x_\invest \to -\infty$, and $\pi_F^\invest(x_\invest; g(x_\invest), x_\invest)$ $\to \infty$ as $x \to \infty$ by Lemma \ref{lemma_payoff_ext}. Thus, by continuity, there must exists $x_\invest^*$ such that $\pi_F^\invest(x_\invest^*;$ $ g(x_\invest^*), x_\invest^*) = 0$. Let $x_L^* = g(x_\invest^*)$. It follows, by Lemma \ref{lemma_payoff_ext}, that there exists $x_\notinvest^*$ such that $\pi_F^\notinvest(x_\notinvest^*; g(x_\invest^*), x_\notinvest^*) = 0$.

\vskip 0.5 \baselineskip
\noindent \underline{Step 2}.
We now prove that $x_\invest^* < x_L^*$. By way of contradiction, assume $x_\invest^* \geq x_L^*$. Then we have $(x_\invest^* - x_L^*)/\sigma \geq 0$ and hence $x_L^* = \Phi((x_\invest^* - x_L^*)/\sigma) \geq 1/2$. It follows that
\begin{align*}
    x_\invest^* + \frac{\sigma_F^2}{\sigma} \lambda \left( \frac{x_\invest^* - x_L^*}{\sigma}  \right) - \frac{n-1}{n} S \left( \frac{x_\invest^* - x_L^*}{\sigma} \right) 
    \geq x_L^* - \frac{n-1}{2n} \geq \frac{1}{2n} > 0,
\end{align*}
which leads to a contradiction. The first inequality is due to $\lambda > 0$, $S < 1/2$, and the assumption that $x_\invest^* \geq x_L^*$. The second inequality is given by $x_L^* \geq 1/2$. Thus, it must be that $x_\invest^* < x_L^*$.

\vskip 0.5 \baselineskip
\noindent \underline{Step 3}.
Let $\widecheck{x}_L^*$ and $\widecheck{x}_\invest^*$ be the limits of $x_L^*$ and $x_\invest^*$, respectively, as $\sigma_L \to 0$ while keeping the ratio $\sigma_F/\sigma_L = \gamma$ fixed.
Step 2 implies that $(x_\invest^* - x_L^*)/\sigma$ can approach either a constant $k \leq 0$ or $-\infty$ as $\sigma_L \to 0$. We show in both cases, $\widecheck{x}_L^*$ is strictly less than $(n-1)/(2n)$ and so is $\widecheck{x}_\invest^*$.
In the former case, $x_L^*$ and $x_\invest^*$ must have the same limit; otherwise $(x_\invest^* - x_L^*)/\sigma \to -\infty$ leading to a contradiction. But since
\begin{equation*}
    \frac{\sigma_F^2}{\sigma} \lambda \left( \frac{x_\invest^* - x_L^*}{\sigma} \right) - \frac{n-1}{n} S \left( \frac{x_\invest^* - x_L^*}{\sigma} \right) \to 0 \cdot \lambda(k) - \frac{n-1}{n} S(k) = - \frac{n-1}{n}S(k),
\end{equation*}
we have $\widecheck{x}_\invest^* = ((n-1)/n)S(k) < (n-1)/(2n)$ by Lemma \ref{lemma_s_func}, and so does $\widecheck{x}_L^*$. In the latter case, $\widecheck{x}_L^* = 0 < (n-1)/(2n)$ because $\Phi((x_\invest^* - x_L^*)/\sigma) \to 0$.

Now consider a function $\widehat{x}_\invest = \widehat{x}_\invest(\sigma_L)$ that approaches $(n-1)/(2n)$ as $\sigma_L \to 0$. Observe that
\begin{equation*}
    \widehat{x}_\invest + \frac{\sigma_F^2}{\sigma} \lambda \left( \frac{\widehat{x}_\invest - x_L^*}{\sigma} \right) - \frac{n-1}{n} S \left( \frac{\widehat{x}_\invest - x_L^*}{\sigma} \right) \to 0
\end{equation*}
because $(\widehat{x}_\invest - x_L^*)/\sigma \to \infty$. This means that $(n-1)/(2n)$ is a solution to $\pi_F^\invest(x_\invest; $ $\widecheck{x}_L^*, x_\invest) = 0$. In other words, given $x_L^*$, both $\widecheck{x}_\invest^*$ and $(n-1)/(2n)$ are solutions to $\pi_F^\invest(x_\invest; x_L^*, x_\invest) = 0$ in the limit as $\sigma_L \to 0$. 
If we can show that the threshold of the leader's best response given $(n-1)/(2n)$ is different from $\widecheck{x}_L^*$, then we are done.

Let $\zeta$ be the limit of $g(\widehat{x}_\invest)$ as $\sigma_L \to 0$. We show next that $\zeta = (n-1)/(2n)$. Suppose that $\zeta < (n-1)/(2n)$, then $(\widehat{x}_\invest - g(\widehat{x}_\invest))/\sigma \to \infty$ as $\sigma_L \to 0$. This implies that $\zeta = 1 > (n-1)/(2n)$, which gives a contradiction to the assumption that $\zeta < (n-1)/(2n)$. Suppose, on the other hand, that $\zeta > (n-1)/(2n)$. This again leads to a contradiction because $\Phi((\widehat{x}_\invest - g(\widehat{x}_\invest))/\sigma) \to 0$ as $\sigma_L \to 0$. Thus, $\zeta = (n-1)/(2n)$. We can now conclude that both actions are rationalizable for leader types $x_L \in (\widecheck{x}_L^*, (n-1)/(2n))$ and for follower types $x \in (\widecheck{x}_\invest^*, (n-1)/(2n))$ under $h = \invest$. The proof is complete. \qed

%% file: appendix_b.tex
\section*{Appendix B: Unique Monotone Equilibrium}
In this appendix, we show that the game has a unique equilibrium when we restrict attention to monotone strategies. Although the result is standard in the literature, it is included for completeness.

Let $s_L: \Theta \to \Action_L$ be the leader's strategy. 
The leader is said to follow a \emph{monotone strategy} if her strategy takes the form:
\begin{equation*}
    s_L(\theta) = \begin{cases}
        \invest & \text{if $\theta > \thetazerohat$} \\
        \notinvest & \text{if $\theta \leq \thetazerohat$}
    \end{cases}.
\end{equation*}
A strategy for any follower $j$ is a mapping $s_j: X_j \times \Action_L \to \Action_j$. Follower $j$'s strategy is monotone if
\begin{equation*}
    s_j(x_j, h) = \begin{cases}
        \invest & \text{if $x_j > \xhhat$} \\
        \notinvest & \text{if $x_j \leq \xhhat$}
    \end{cases}.
\end{equation*}
A \emph{monotone equilibrium} is a symmetric perfect Bayesian equilibrium in monotone strategies with thresholds $(\theta_L^*, x_\invest^*, x_\notinvest^*)$.

\begin{lemma} \label{lemma_monotone_b}
There exists a monotone equilibrium with thresholds $\thetazerostar = 0$, $\xistar = -\infty$, and $\xnstar = \infty$.
\end{lemma}

\begin{proof}
    Fix a follower type $x$. Suppose that the leader uses threshold $\thetazerostar = 0$ and other followers use thresholds $\xistar = -\infty$ and $\xnstar = \infty$. If the leader exerts effort, then type $x$'s payoff yields
\begin{equation*}
    \pi_F^\invest(x; \thetazerostar, \xistar) = \EE_{\theta \sim \Psi^\invest(\cdot; \, x, 0)} \left[\theta\right] > 0.
\end{equation*}
This means that all types $x$ will exert effort under history $h = \invest$. Thus, follower $j$'s best response is a monotone strategy with threshold $\xistar = -\infty$. In contrast, if the leader does not exert effort, then the payoff for type $x$ is
\begin{equation*}
    \pi_F^\notinvest(x; \thetazerostar, \xnstar) = \EE_{\theta \sim \Psi^\notinvest(\cdot; \,x, 0)} [\theta] - 1< 0.
\end{equation*}
Thus, under history, $h = \notinvest$, follower $j$ will best respond by using a monotone strategy with threshold $\xnstar = \infty$.

Consider now type $\theta$ of the leader. Since all followers will invest if they see the leader invests, investing generates a payoff of $\theta$ for type $\theta$. Therefore, type $\theta$ invests if and only if $\theta > 0$. In other words, the leader will best respond by choosing threshold $\thetazerostar = 0$. The proof is complete.
\end{proof}

\begin{proposition}
There is no monotone equilibrium other than the one given in Lemma \ref{lemma_monotone_b}.

\end{proposition}
\begin{proof}
By way of contradiction, suppose that $\thetazerostar$ and $\xistar$ are the equilibrium thresholds. Then they must solve the indifference conditions
\begin{equation} \label{app_b_leader_eq}
    \pi_L(\thetazerostar; \xistar) = \thetazerostar - \Phi\left(\frac{\xistar - \thetazerostar}{\sigma_F}\right) = 0 \tag{B.1}
\end{equation}
and
\begin{equation} \label{app_b_follower_eq}
    \pi_F^\invest(x_\invest^*; \theta_L^*, x_\invest^*) =  \EE_{\theta \sim \Psi^\invest(\cdot; \,\xistar, \thetazerostar)} \left[ \theta - \frac{n-1}{n}\Phi\left(\frac{\xistar - \thetazerostar}{\sigma_F}\right) \right] = 0. \tag{B.2}
\end{equation}
By Equations (\ref{rank_belief}) and (\ref{eq_eqm_follower}) we can write (\ref{app_b_follower_eq}) as
\begin{equation} \label{app_b_follower_eq_2}
    \xistar + \sigma_F \lambda\left(\frac{\xistar - \thetazerostar}{\sigma_F}\right) = \frac{n-1}{2n}\Phi\left(\frac{\xistar - \thetazerostar}{\sigma_F}\right). \tag{B.3}
\end{equation}
Subtracting Equation (\ref{app_b_leader_eq}) from Equation (\ref{app_b_follower_eq_2}) yields
\begin{equation} \label{app_eq_diff}
    \xistar - \thetazerostar + \sigma_F \lambda\left( \frac{\xistar - \thetazerostar}{\sigma_F} \right) = - \frac{n+1}{2n} \Phi\left(\frac{\xistar - \thetazerostar}{\sigma_F}\right). \tag{B.4}
\end{equation}
Note that $x + \lambda(x)$ is increasing in $x$ with $\lim_{x \to -\infty} x + \lambda(x) = 0$, and hence $x + \lambda(x) > 0$ for all $x$. This implies that the left-hand side of (\ref{app_eq_diff}) is positive. But since the right-hand side of (\ref{app_eq_diff}) is negative, this leads to a contradiction.
\end{proof}

%% file: appendix_c.tex
\section*{Appendix C: Log-concave Noises}

In this appendix, we extend the main model in Section 3 by considering noises that have log-concave densities. Formally, $(\varepsilon_i)_{i \in N}$ are independently drawn from distribution $F$, which has a positive continuous density $f$ on the entire real line. We assume, in addition, that $f$ is strictly log-concave and symmetric about zero. Common distributions, such as Gaussian, Laplace, and logistic distributions with mean zero, satisfy these assumptions.

Suppose that the leader uses a monotone strategy with threshold $z \in \RR$. Then a follower with type $x$ has the following posterior density about $\theta$:
\begin{equation*}
    \psi^h(\theta; x, z) = \begin{cases}
    \frac{\frac{1}{\sigma_F}f\left(\frac{x - \theta}{\sigma_F}\right)}{F\left(\frac{x - z}{\sigma_F}\right)}\one(\theta > z) & \text{if $h = \invest$} \\
    ~ & ~ \\
    \frac{\frac{1}{\sigma_F}f\left(\frac{x - \theta}{\sigma_F}\right)}{1 - F\left(\frac{x - z}{\sigma_F}\right)}\one(\theta \leq z) & \text{if $h = \notinvest$}
    \end{cases}.
\end{equation*}
Let $\Psi^h(\cdot; \,x, z)$ be the corresponding CDF.

\begin{lemma} \label{app_c_fosd}
$\Psi^h(\cdot; \,x, z)$ is strictly increasing in $x$ and $z$ in the sense of strict first-order stochastic dominance.
\end{lemma}
\begin{proof}
We only prove the case $h = \invest$, for the proof of the other case, is similar.
Fix $z$ and $\theta > z$, and let $x < x'$. Note that
\begin{equation*}
    \frac{ \psi^\invest (\theta; \,x', z) }{ \psi^\invest (\theta; \,x', z) } = \frac{ F\left(\frac{x - z}{\sigma_F}\right) }{ F\left(\frac{x' - z}{\sigma_F}\right) } \cdot \frac{ f\left(\frac{x' - \theta}{\sigma_F}\right) }{ f\left(\frac{x - \theta}{\sigma_F}\right) }.
\end{equation*}
Since $f$ is logconcave, $f((x'-\theta)/\sigma_F)/f((x-\theta)/\sigma_F)$ is weakly increasing in $\theta$; that is, $\Psi^\invest(\cdot; \,x', z)$ dominates $\Psi^\invest(\cdot; \,x, z)$ in the monotone likelihood ratio order. This follows from the equivalence between log concavity and P\'{o}lya Frequency of order 2 (See, for example, Proposition 1 in \cite{an_1998} or Proposition 2.3 in \cite{saumard_wellner_2014}). 
Thus, $\Psi^\invest(\cdot; \,x', z)$ first-order stochastically dominates $\Psi^\invest(\cdot; \,x, z)$.

The claim that $\Psi^\invest(\cdot; \,x, z)$ is strictly increasing in the first-order stochastic dominance sense is because
\begin{equation*}
    \Psi^\invest(\theta; \,x, z) = \frac{1}{ F\left(\frac{x-z}{\sigma_F}\right)  }\int_z^\theta \frac{1}{\sigma_F} f\left(\frac{x-t}{\sigma_F}\right)  \dd t = 1 - \frac{ F\left(\frac{x-\theta}{\sigma_F}\right) }{ F\left(\frac{x-z}{\sigma_F}\right) }
\end{equation*}
is strictly increasing in $z$.
\end{proof}

Let $\eta = \lim_{x \to -\infty} F(x)/f(x)$. It is worth noting that $\eta$ is the scale parameter if the noises follow a Laplace or logistic distribution.
\begin{lemma} \label{app_c_exp}
The posterior expectations $\EE_{\theta \sim \Psi^h(\cdot; \,x, z)}[\theta]$ are strictly increasing in $x$ and $z$. Moreover,
\begin{equation*}
    \lim_{x \to -\infty} \EE_{\theta \sim \Psi^h(\cdot; \,x, z)}[\theta] = \begin{cases}
    \sigma_F \eta + z & \text{if $h = \invest$} \\
    - \infty & \text{if $h = \notinvest$}
    \end{cases},
\end{equation*}
and 
\begin{equation*}
    \lim_{x \to \infty} \EE_{\theta \sim \Psi^h(\cdot ; \,x, z)}[\theta] = \begin{cases}
    \infty & \text{if $h = \invest$} \\
    -\sigma_F \eta + z & \text{if $h = \notinvest$}
    \end{cases}.
\end{equation*}
\end{lemma}
\begin{proof}
(Part 1) The first part of this lemma is given by Lemma \ref{app_c_fosd}.

\noindent (Part 2) Consider first history $h = \invest$.
By a change of variable, we have
\begin{equation*}
    \EE_{\theta \sim \Psi^\invest(\cdot; \,x, z)}[\theta]  = \frac{1 }{F\left(\frac{x - z}{\sigma_F}\right)} \int_z^\infty  \frac{t}{\sigma_F} f\left(\frac{x - t}{\sigma_F}\right) \dd t = \sigma_F \delta\left( \frac{x - z}{\sigma_F} \right) + z.
\end{equation*}
where
\begin{equation*}
    \delta(u) = u - \frac{ \int_{-\infty}^u t f(t) \dd t }{ F(u) }.
\end{equation*}
It is clear that $\lim_{u \to \infty} \delta(u) = \infty$, and hence
\begin{equation*}
    \lim_{x \to \infty} \EE_{\theta \sim \Psi^h(\cdot ; \,x, z)}[\theta] = \infty.
\end{equation*}
Since $f$ is log-concave, it has a light right tail; that is,
\begin{equation*}
    \lim_{x \to \infty} \frac{ f(x) }{ \ee^{-cx} } = 0
\end{equation*}
for some $c > 0$.\footnote{~See, for example, Corallary 1 in \cite{an_1998}.} Thus, by symmetry (about zero),
\begin{equation*}
    \lim_{x \to -\infty} x F(x) = \lim_{x \to -\infty} \frac{ F(x) }{ x^{-1} } = \lim_{x \to -\infty} \frac{ f(x) }{ -x^{-2} } = \lim_{x \to -\infty} \frac{ f(x) }{ \ee^{cx} } \cdot \frac{ \ee^{cx} }{ -x^{-2} } = 0.
\end{equation*} 
Integrating by parts now gives 
\begin{equation*}
    \delta(u) = \frac{ \int_{-\infty}^u F(t) \dd t }{ F(u) }.
\end{equation*}
It follows that $\lim_{u \to -\infty} \delta(u) = \eta$. Thus,
\begin{equation*}
    \lim_{x \to -\infty} \EE_{\theta \sim \Psi^h(\cdot ; \,x, z)}[\theta] = \sigma_F \eta + z.
\end{equation*}

The proof for history $h = \notinvest$ is similar.
We have
\begin{equation*}
    \EE_{\theta \sim \Psi^\notinvest(\cdot; \,x, z)}[\theta] = \frac{1}{ 1 - F\left(\frac{x - z}{\sigma_F}\right) } \int_{-\infty}^z \frac{t}{\sigma_F} f\left( \frac{x - t}{\sigma_F} \right) \dd t = -\sigma_F \varsigma \left(\frac{x - z}{\sigma_F} \right) + z,
\end{equation*}
where
\begin{equation*}
    \varsigma(u) = \frac{ \int_u^\infty t f(t) \dd t }{ 1 - F(u) } - u.
\end{equation*}
Thus, $\lim_{u \to -\infty} \varsigma(u) = \infty$, and
\begin{equation*}
    \lim_{x \to -\infty} \EE_{\theta \sim \Psi^\notinvest(\cdot; \,x, z)}[\theta] = -\infty.
\end{equation*}
Since $\lim_{u \to \infty} u \left( 1 - F(u) \right) = 0$ (because $f$ is light-tailed),
\begin{equation*}
    \varsigma(u) = \frac{\int_u^\infty [1 - F(t)] \dd t}{1 - F(u)}
\end{equation*}
by integration by parts. Thus,
\begin{equation*}
    \lim_{u \to \infty} \varsigma(u) = \lim_{u \to \infty} \frac{1 - F(u)}{f(u)} = \lim_{u \to \infty} \frac{ F(-u) }{ f(-u) } = \eta,
\end{equation*}
which implies that
\begin{equation*}
    \lim_{x \to \infty} \EE_{\theta \sim \Psi^\notinvest(\cdot; \,x, z)}[\theta] = -\sigma_F \eta + z.
\end{equation*}
The proof is complete.
\end{proof}

Suppose, in addition, that a follower with type $x$ believes that other followers use monotone strategies with threshold $x_h$ under history $h$, then his payoff under history $h$
\begin{equation*}
    \pi_F^h(x; z, x_h) = \EE_{\theta \sim \Psi^h(\cdot; \,x, z)} \left[ \theta - \frac{n-1}{n}F\left( \frac{x_h - \theta}{\sigma_F} \right) \right] - \frac{\chi_\notinvest}{n}
\end{equation*}
has the following properties:

\begin{lemma} \label{app_c_fol_payoff}
Type $x$'s payoffs $\pi_F^h(x; z, x_h)$ are strictly increasing in $x$ and $z$ but is strictly decreasing in $x_h$. Moreover,
\begin{equation*}
   \lim_{x \to -\infty} \pi_F^h(x; z, x_h) = \begin{cases}
        \sigma_F \eta + z - \frac{n-1}{n}F\left( \frac{x_\invest - z}{\sigma_F} \right) & \text{if $h = \invest$} \\
        -\infty & \text{if $h = \notinvest$}
    \end{cases}
\end{equation*}
and
\begin{equation*}
    \lim_{x \to \infty} \pi_F^h(x; z, x_h) = \begin{cases}
        \infty & \text{if $h = \invest$} \\
        -\sigma_F \eta + z - \frac{1}{n} - \frac{n-1}{n}F\left( \frac{x_\notinvest - z}{\sigma_F} \right) & \text{if $h = \notinvest$}
    \end{cases}
\end{equation*}
\end{lemma}
\begin{proof}
(Part 1) Note that $\theta - ((n-1)/n)F((x_h - \theta)/\sigma_F)$ is strictly increasing in $\theta$; therefore Lemma \ref{app_c_fosd} implies that $\pi_F^h(x; z, x_h)$ is strictly increasing in $x$ and $z$. But since $\theta - ((n-1)/n)F((x_h - \theta)/\sigma_F)$ is strictly decreasing in $x_h$, so is $\pi_F^h(x; z, x_h)$.

(Part 2)
It follows immediately from Lemma \ref{app_c_exp} that $\pi_F^\invest(x; z, x_\invest) \to \infty$ as $x \to \infty$
and $\pi_F^\notinvest(x; z, x_\notinvest) \to -\infty$ as $x \to -\infty$
because $F((x_h - \theta)/\sigma_F)$ is bounded.

Now under $h = \invest$, a change of variable $u = (x - \theta)/\sigma_F$ gives that
\begin{align*}
    \EE_{\theta \sim \Psi^\invest(\cdot;\,x, z)}\left[ F\left(\frac{x_\invest - \theta}{\sigma_F}\right)\right] & = \frac{1}{F\left(\frac{x - z}{\sigma_F}\right)} \int_z^\infty F\left(\frac{x_\invest - \theta}{\sigma_F}\right) \frac{1}{\sigma_F}f\left(\frac{x - \theta}{\sigma_F}\right) \dd \theta \\
    & = \frac{1}{F\left(\frac{x - z}{\sigma_F}\right)} \int_{-\infty}^\frac{x-z}{\sigma_F} f(u) F\left(\frac{x_\invest - x}{\sigma_F} + u \right) \dd u.
\end{align*}
By L'H\^{o}pital's rule, the above expectation has the following limit:
\begin{equation*}
    \lim_{x \to -\infty} \EE_{\theta \sim \Psi^\invest(\cdot;\,x, z)}\left[ F\left(\frac{x_\invest - \theta}{\sigma_F}\right)\right] = F \left( \frac{x_\invest - z}{\sigma_F} \right).
\end{equation*}
Thus, by Lemma \ref{app_c_exp}, we can conclude that
\begin{equation*}
    \lim_{x \to -\infty} \pi_F^\invest(x; z, x_\invest) = \sigma_F \eta + z - \frac{n-1}{n} F \left( \frac{x_\invest - z}{\sigma_F} \right).
\end{equation*}
Similarly, under $h = \notinvest$,
\begin{align*}
    \lim_{x \to \infty} \EE_{\theta \sim \Psi^\notinvest(\cdot;\,x, z)}\left[ F\left(\frac{x_\notinvest - \theta}{\sigma_F}\right)\right] & = \lim_{x \to \infty} \frac{1}{F\left( \frac{z-x}{\sigma_F} \right)} \int_{-\infty}^z F \left( \frac{x_\notinvest - \theta}{\sigma_F} \right) \frac{1}{\sigma_F} f\left( \frac{x-\theta}{\sigma_F} \right) \dd \theta \\
    & = \lim_{x \to \infty} \frac{1}{F\left( \frac{z-x}{\sigma_F} \right)} \int_{\frac{x-z}{\sigma_F}}^\infty f(u) F \left( \frac{x_\notinvest - x}{\sigma_F} + u \right) \dd u \\
    & = F \left( \frac{x_\notinvest - z}{\sigma_F} \right).
\end{align*}
Thus, 
\begin{equation*}
    \lim_{x \to \infty} \pi_F^\notinvest (x; z, x_\notinvest) = -\sigma_F \eta + z - \frac{1}{n} - \frac{n-1}{n} F \left( \frac{x_\notinvest - z}{\sigma_F} \right),
\end{equation*}
as desired.
\end{proof}

Let $(\thetalow^k, \thetaup^k, \xiup^k, \xilow^k, \xnup^k, \xnlow^k)_{k=0}^\infty$ be the six $\Delta$-rationalizable sequences, where $\thetalow^0 = \xilow^0 = \xnlow^0 = -\infty$ and $\thetaup^0 = \xiup^0 = \xnup^0 = \infty$.
Let $\thetalow$, $\thetaup$, $\xilow$, $\xiup$, $\xnlow$, $\xnup$ be their limits, respectively, as $k \to \infty$. If one sequence diverges, then its limit is either $-\infty$ or $\infty$. 

\begin{proposition} \label{app_c_rat_seq}
    If $\sigma_F \geq (n\eta)^{-1}(n-1)$, then the above sequences have the following properties: \\
    (a) $(\thetalow^k)_{k=0}^\infty$ is such that $\thetalow^k = \thetalow = 0$ for all $k \geq 1$; \\
    (b) $(\thetaup^k)_{k=0}^\infty$ is decreasing and  such that $\thetaup^k = \thetaup = 0$ for all $k \geq 2$; \\
    (c) $(\xilow^k)_{k=0}^\infty$ is such that $\xilow^k = \xilow = -\infty$ for all $k \geq 0$; \\
    (d) $(\xiup^k)_{k=0}^\infty$ is decreasing and such that $\xiup^k = \xiup -\infty$ for all $k \geq 1$; \\
    (e) $(\xnlow^k)_{k=0}^\infty$ is increasing and such that $\xnlow^k = \xnlow = \infty$ for all $k \geq 1$; \\
    (f) $(\xnup^k)_{k=0}^\infty$ is such that $\xnup^k = \xnup = \infty$ for all $k \geq 0$.
\end{proposition}
\begin{proof}
\textit{Round $k = 1$}: First note that the best-case and worst-case payoffs to leader $\theta$ are given by
\[
\pi_L(\theta; \xilow^0) = \theta - F \left( \frac{\xilow^0 - \theta}{\sigma_F} \right) = \theta
\]
and 
\[
\pi_L(\theta; \xiup^0) = \theta - F \left( \frac{\xiup^0 - \theta}{\sigma_F} \right) = \theta -1,
\]
respectively.
This implies that $\thetalow^1 = 0$ and $\thetaup^1 = 1$. That is, it is dominant for all leader types below $\thetalow^1$ to take $a_L = \notinvest$ and dominant for all leader types above $\thetaup^1$ to choose $a_L = \invest$.

Given $\thetalow^1$ and $\thetaup^1$. By Lemma \ref{app_c_fol_payoff}, follower $x$'s best-case payoff under $h = \invest$ is
\[
\pi_F^\invest(x; \thetaup^1, \xilow^0) = \EE_{\theta \sim \Psi^\invest(\cdot; x, \thetaup^1)}\left[\theta - \frac{n-1}{n} F\left( \frac{\xilow^0 - \theta}{\sigma_F} \right) \right] = \EE_{\theta \sim \Psi^\invest(\cdot; x, \thetaup^1)} [\theta].
\]
Since $\lim_{x \to -\infty} \EE_{\theta \sim \Psi^\invest(\cdot; x, \thetaup^1)} [\theta] = \sigma_F \eta + 1 > 0$ by Lemma \ref{app_c_exp}, we have $\xilow^1 = -\infty$, meaning that there is no follower type for whom action $\invest$ is strictly dominated. Analogously, the worst-case payoff to follower $x$ yields
\begin{align*}
  \pi_F^\invest(x; \thetalow^1, \xiup^0) & = \EE_{\theta \sim \Psi^\invest(\cdot; x, \thetalow^1)}\left[\theta - \frac{n-1}{n} F\left( \frac{\xiup^0 - \theta}{\sigma_F} \right) \right] \\
  & = \EE_{\theta \sim \Psi^\invest(\cdot; x, \thetalow^1)} [\theta] - \frac{n-1}{n},
\end{align*}
which has the limit of $\sigma_F \eta - (n-1)/n$ as $x \to -\infty$. But since $\sigma_F \geq (n-1)/(n\eta)$, it follows that
\[
\lim_{x \to -\infty} \pi_F^\invest(x; \thetalow^1, \xiup^0) = \sigma_F \eta - \frac{n-1}{n} \geq 0.
\]
Thus, $\xiup^1 = -\infty$ and it is dominant for all follower types to take action $\invest$.

Under $h = \notinvest$, the worst-case payoff to follower $x$ is
\begin{align*}
    \pi_F^\notinvest(x; \thetalow^1, \xnup^0) & = \EE_{\theta \sim \Psi^\notinvest(\cdot; x, \thetalow^1)} \left[\theta - \frac{n-1}{n} F\left( \frac{\xnup^0 - \theta}{\sigma_F} \right) \right] - \frac{1}{n} \\ 
    & = \EE_{\theta \sim \Psi^\notinvest(\cdot; x, \thetalow^1)} [\theta] - 1.
\end{align*}
By Lemma \ref{app_c_exp}, we have
\[
\lim_{x \to \infty} \pi_F^\notinvest(x; \thetalow^1, \xnup^0) = - \sigma_F \eta - 1 < 0.
\]
Therefore $\xnup^1 = \infty$; i.e., there exists no follower type for whom $\invest$ is a dominant action. Follower $x$'s best-case payoff exhibits 
\begin{align*}
    \pi_F^\notinvest(x; \thetaup^1, \xnlow^0) & = \EE_{\theta \sim \Psi^\notinvest(\cdot; x, \thetaup^1)} \left[\theta - \frac{n-1}{n} F\left( \frac{\xnlow^0 - \theta}{\sigma_F} \right) \right] - \frac{1}{n} \\
    & = \EE_{\theta \sim \Psi^\notinvest(\cdot; x, \thetaup^1)} [\theta] - \frac{1}{n}
\end{align*}
and
\begin{align*}
    \lim_{x \to \infty} \pi_F^\notinvest(x; \thetaup^1, \xnlow^0) = -\sigma_F \eta + 1 - \frac{1}{n} \leq -\left(\frac{n-1}{n \eta}\right) \eta + \frac{n-1}{n} = 0. 
\end{align*}
Thus we have $\xnlow^1 = \infty$; i.e., $\notinvest$ is a dominant action for all follower types.

\textit{Round $k = 2$}: Given $\xilow^1 = \xiup^1 = -\infty$, leader $\theta$'s worst-case and best-case payoffs coincide:
\begin{equation*}
    \pi_L(\theta; \xiup^1) = \theta - F \left( \frac{\xiup^1 - \theta}{\sigma_F} \right) = \theta.
\end{equation*}
This immediately implies that $\thetalow^2 = \thetaup^2 = 0$.

Given $\thetalow^2$ and $\thetaup^2$, follower $x$'s worst-case and best-case payoffs, under history $h = \invest$, also coincide because $\thetalow^2 = \thetaup^2 = 0$ and $\xilow^1 = \xiup^1 = -\infty$. Moreover, by Lemma \ref{app_c_exp},
\begin{equation*}
    \lim_{x \to -\infty} \pi_F^\invest(x; \thetaup^2, \xilow^2) = \sigma_F \eta > 0.
\end{equation*}
Thus, $\xilow^2 = \xiup^2 = -\infty$. One can show analogously that $\xnlow^2 = \xnup^2 = \infty$ because $\thetalow^2 = \thetaup^2 = 0$ and $\xnlow^1 = \xnup^1 = \infty$.

Now we can conclude that (a)-(f) hold true via a simple induction argument.
\end{proof}

\begin{corollary}
If $\sigma_F \geq (n\eta)^{-1}(n-1)$, then the game has a unique $\Delta$-rationalizable strategy profile in which all followers imitate the leader's choice of action.
\end{corollary}

Proposition \ref{app_c_rat_seq} says that sequences $(\thetalow^k)_{k=0}^\infty$, $(\xilow^k)_{k=0}^\infty$, and $(\xnup^k)_{k=0}^\infty$ always converge. However, it is 
worth noting that we have provided a strong sufficient condition so that a unique $\Delta$-rationalizable behavior is achieved in Round 2 (i.e., 
the other three sequences to converge in Round 2). In fact, 
if the following condition holds:
\begin{equation*}
    \max\left\{ \iota_\invest^{k}(\sigma_F), \iota_\notinvest^{k}(\sigma_F) \right\} \leq \sigma_F < \min \left\{ \iota_\invest^{k-1}(\sigma_F), \iota_\notinvest^{k-1}(\sigma_F)  \right\},
\end{equation*}
where
\begin{equation*}
\iota_\invest^k(\sigma_F) = \frac{n-1}{n\eta} F\left( \frac{\xiup^{k-1}}{\sigma_F} \right)
\end{equation*}
and
\begin{equation*}
\iota_\notinvest^k(\sigma_F) = \frac{1}{\eta}\left[ \thetaup^k - \frac{1}{n} - \frac{n-1}{n} F \left( \frac{\xnlow^{k-1} - \thetaup^k}{\sigma_F} \right) \right],
\end{equation*}
then the six sequences will converge to the unique $\Delta$-rationalizable profile in Round $k + 1,~ k \geq 1$.
The next proposition shows that there exist multiple $\Delta$-rationalizable profiles when $\sigma_F \to 0$.

\begin{proposition}
    In the limit as $\sigma_F \to 0$, we have $\thetaup \to (n-1)/(2n)$, $\xiup \to (n-1)/(2n)$, and $\xnlow \to \infty$.
\end{proposition}
\begin{proof}
    The proof is identical to that of Proposition \ref{prop_limit}.
\end{proof}